\documentclass[aps,pra,twocolumn,floatfix]{revtex4-1}
\usepackage{amsmath,amssymb,amsthm,mathtools}
\usepackage{graphicx}
\usepackage{booktabs}
\usepackage[dvipsnames]{xcolor}
\usepackage{enumitem}
\usepackage[colorlinks=true,linkcolor=NavyBlue,citecolor=BrickRed,urlcolor=NavyBlue]{hyperref}

\newcommand{\R}{\mathbb{R}}
\newcommand{\C}{\mathbb{C}}
\newcommand{\hb}{\hbar}
\newcommand{\Dc}{\Delta_{c}}
\newcommand{\Da}{\Delta_{a}}
\newcommand{\dx}{\Dc x}
\newcommand{\dpp}{\Dc p}
\newcommand{\eps}{\varepsilon}

\newcommand{\norm}[1]{\left\lVert #1\right\rVert}
\newcommand{\ip}[2]{\left\langle #1,#2\right\rangle}
\newcommand{\ket}[1]{\lvert #1\rangle}
\newcommand{\bra}[1]{\langle #1\rvert}
\newcommand{\ev}[1]{\langle #1\rangle}
\newcommand{\dd}{\,\mathrm{d}}
\newcommand{\Fou}{\mathcal{F}}
\newcommand{\Tr}{\operatorname{Tr}}
\DeclareMathOperator{\supp}{supp}

\theoremstyle{plain}
\newtheorem{theorem}{Theorem}
\newtheorem{proposition}{Proposition}
\newtheorem{lemma}{Lemma}
\newtheorem{corollary}{Corollary}
\theoremstyle{definition}
\newtheorem{definition}{Definition}

\newtheorem{conjecture}{Conjecture}

\begin{document}

\title{Position--momentum uncertainty relation for the average confidence width}

\author{Shengjun Wu}
\email{sjwu@nju.edu.cn}
\affiliation{National Laboratory of Solid State Microstructures and School of Physics,  Collaborative Innovation Center of Advanced Microstructures, Nanjing University, Nanjing 210093, China.}

\date{\today}

\begin{abstract}
We introduce the \emph{average confidence width}
$\Da x=\int_0^1\dx(\theta)\dd\theta$: the confidence width $\dx(\theta)$---the measure of
the smallest region of position space carrying a fraction $\theta$ of the probability---averaged
over all confidence levels. It is the first moment of the decreasing rearrangement of the
position density, an $L^1$ (mean-absolute-deviation) measure of localization, and the product
$\Da x\,\Da p$ is dilation invariant. We study the uncertainty relation
$\Da x\,\Da p\ge c^\ast\hb$, valid for pure and mixed states. A mean--entropy argument
combined with the Bia\l ynicki-Birula--Mycielski relation gives $c^\ast\ge\pi/e$, while the
ground state $\psi_0$ of the Fourier-invariant operator $|x|+|p|$, with energy
$E_0=1.1040744$, gives $c^\ast\le E_0^2\approx1.2190<4/\pi$: Gaussian states, which
saturate the Heisenberg--Kennard and the entropic relations, are \emph{not} optimal here. We
prove that $\psi_0$ is nondegenerate, nodeless, even, self-dual, and symmetric decreasing, with
algebraic tails, and that the sharp hybrid relations $\Da x\cdot2\ev{|p-p_0|}\ge E_0^2\hb$ and
$2\ev{|x-x_0|}\cdot\Da p\ge E_0^2\hb$ hold for all states, so that $\Da x\,\Da p\ge E_0^2\hb$
whenever one of the two densities is symmetric unimodal; we also show that the second variation of
$\Da x\,\Da p$ at $\psi_0$ is positive definite apart from the symmetry directions. Extensive numerical
minimization, including expansions in
coherent states, never finds a product below $E_0^2\hb$ and consistently returns $\psi_0$,
supporting the conjecture $c^\ast=E_0^2$. In
two dimensions the corresponding relation is sharp and saturated by Gaussians, which
identifies the dimension as the origin of the non-Gaussian optimizer in one dimension. Finally, we
discuss other complementary pairs and the precise relation of the width to interferometric
predictability and visibility.
\end{abstract}

\maketitle

\section{Introduction}\label{sec:intro}

Uncertainty relations quantify the impossibility of preparing a state that is sharp in two
incompatible observables, and each relation rests on a particular notion of ``spread.''
The variance relations of Heisenberg, Kennard, and
Robertson~\cite{Heisenberg1927,Kennard1927,Robertson1929,Dodonov2025} weight the tails of a
distribution quadratically; the entropic relations~\cite{Hirschman1957,Beckner1975,BBM1975,
MaassenUffink1988,Wu2009,ColesRMP2017,Wang2019,Huang2021} use the Shannon entropy of the
outcome distribution; the support and concentration principles of Landau--Pollak and
Donoho--Stark~\cite{LandauPollak1961,DonohoStark1989,FollandSitaram1997} measure the size
of the sets on which a function and its Fourier transform are concentrated; and
majorization relations~\cite{Partovi2011,FriedlandGheorghiuGour2013,
PuchalaRudnickiZyczkowski2013} order the distributions themselves. For position and
momentum the variance relation $\sigma_x\sigma_p\ge\hb/2$ and the entropic relation of
Bia\l ynicki-Birula and Mycielski (BBM)~\cite{BBM1975} are both saturated by Gaussian states,
and one might expect Gaussians to be optimal for every reasonable measure of spread. The
purpose of this paper is to study a natural, operationally defined measure for which this
is \emph{not} the case.

A measure-theoretic notion of localization was introduced recently in
Ref.~\cite{LinConfidence2026}, under the name confidence uncertainty; here we call it the
\emph{confidence width} $\dx(\theta_x)$: the size of the smallest position region
carrying a fraction $\theta_x$ of the probability. It
quantifies localization operationally---the total length one must cover to be
$\theta_x$-confident of finding the particle---but it depends on the arbitrary level
$\theta_x$. Averaging over all levels removes this dependence and yields the object of
this paper, the \emph{average confidence width} (ACW) $\Da x=\int_0^1\dx(\theta_x)\dd\theta_x$
(and likewise $\Da p$), the area under the concentration curve. It has a clean
analytic identity: $\Da x$ is the first moment of the decreasing rearrangement of
the position density, equivalently twice the mean absolute deviation for symmetric unimodal
densities---an $L^1$ measure of localization rather than an $L^2$ one. It also carries an
information-theoretic meaning: it is bounded below by the exponential of the differential
entropy, $\Da x\ge e^{h(x)-1}$, and for a discrete spectrum it coincides with Massey's
guesswork~\cite{Massey1994,Arikan1996}.

The product $\Da x\,\Da p$ is dilation invariant, so an uncertainty relation must take the
scale-free form $\Da x\,\Da p\ge c^\ast\hb$ with a universal optimal constant $c^\ast$.
Our results on this relation are the following.
(i) Two rigorous lower bounds: an elementary one obtained from the confidence bound of
Ref.~\cite{LinConfidence2026}, and the entropic bound $c^\ast\ge\pi/e\approx1.1557$.
(ii) A rigorous upper bound $c^\ast\le E_0^2$, where $E_0=1.1040744$ is the ground-state
energy of the Fourier-invariant operator $|x|+|p|$. Since $E_0^2\approx1.2190<4/\pi\approx1.2732$,
the Gaussian value, \emph{Gaussian states are not the minimizers}. The operator $|x|+|p|$ is
known from the work of Werner on joint measurements~\cite{Werner2004} and of Busch, Lahti, and
Werner~\cite{BuschLahtiWerner2014}, where its ground-state energy determines the sharp measurement and
preparation uncertainty relations for mean
absolute deviations, $\ev{|x-x_0|}\ev{|p-p_0|}\ge(E_0^2/4)\hb\approx0.3047\,\hb$; in
Ref.~\cite{BuschLahtiWerner2014} the nondegeneracy of
its ground state was taken for granted on numerical grounds. We prove that the ground
state $\psi_0$ is nondegenerate, nodeless, even, equal to its own Fourier transform, and symmetric
decreasing, so that the bound $E_0^2\hb$ is attained by $\psi_0$, and that it has algebraic tails,
bounded above and below by multiples of $|x|^{-3}$ (Theorem~\ref{prop:hamiltonian}).
(iii) Sharp \emph{hybrid} relations, $\Da x\cdot2\ev{|p-p_0|}\ge E_0^2\hb$ and its mirror image,
valid for all states, which strengthen the relation for mean absolute deviations. They imply
$\Da x\,\Da p\ge E_0^2\hb$ for every state with at least
one symmetric unimodal density, however irregular the other density may be
(Theorem~\ref{thm:hybrid}).
(iv) The second variation of $\Da x\,\Da p$ at $\psi_0$ along smooth perturbations is nonnegative and
vanishes only along the symmetry directions (translation, boost, dilation), with a spectral gap: $\psi_0$ is
a nondegenerate critical point, stable at second order (Proposition~\ref{prop:local}).
(v) A monotone self-consistent minimization algorithm, run in Hermite subspaces, on
lattices of coherent states, and on the full discretized Hilbert space, converges from random and
structured starting points to $\psi_0$; the minima approach $E_0^2$ from above, and, apart from lattice-scale
artifacts on the coarsest grid, the only other local minima found are comb-like states with product $\approx\pi/2$. Together with
(iii) and (iv) this strongly supports the conjecture $c^\ast=E_0^2$, although a complete proof
remains open.
(vi) All uncertainty relations above extend to mixed states with the same constants; in addition, the von Neumann entropy $S(\rho)$ enters
through the bound $\Da x\,\Da p\ge(2\pi/e^2)e^{S(\rho)}\hb$, which is stronger than $\pi/e$ for sufficiently mixed states.
(vii) In $d$ dimensions the analogous upper bound is governed by the operator $|x|^d+|p|^d$; for $d=2$ this is the
harmonic oscillator, the entropic bound is sharp, and Gaussians \emph{are} optimal, which
explains the origin of the non-Gaussian optimizer in one dimension.

The relation between uncertainty and wave--particle duality is discussed in Sec.~\ref{sec:compl}, where we also treat other
complementary pairs (rotated quadratures, which supply a continuum of bases mutually
unbiased with both position and momentum), the path qubit of a Mach--Zehnder
interferometer, for which the widths of the three mutually unbiased observables reproduce
predictability and visibility exactly, and the double slit, for which $\Da x$ and $\Da p$ are
monotone functions of predictability and visibility.  Proofs and numerical methods are collected in the appendices.

\textbf{Background.} We use two structural facts established for the confidence
width in \cite{LinConfidence2026}.
For a state with position density $f(x)=\bra{x}\rho\ket{x}$ (for a pure state
$f=|\langle x|\psi\rangle|^2$) and $\theta_x\in[0,1]$,
\begin{equation}\label{eq:defconf}
  \dx(\theta_x)=\inf\Big\{\mu(X):\ \int_X f(x)\dd x\ge\theta_x\Big\},
\end{equation}
with $X\subseteq\R$ Lebesgue measurable and $\mu$ Lebesgue measure; $\dpp(\theta_p)$ is
defined in the same way through the momentum density $g(p)=\bra{p}\rho\ket{p}$, with $\langle p|\psi\rangle=\tfrac1{\sqrt{2\pi\hb}}\int e^{-ipx/\hb}\langle
x|\psi\rangle\dd x$. The two facts are: (1) \emph{(joint localizability)} if
$\theta_x+\theta_p\le1$, then for every $\eps>0$ some state has
$\dx(\theta_x)<\eps$ and $\dpp(\theta_p)<\eps$ simultaneously; and (2) the
\emph{confidence bound}: for $\theta_x+\theta_p>1$,
\begin{equation}\label{eq:thm2}
  \dx(\theta_x)\,\dpp(\theta_p)\ \ge\ 2\pi\hb\Big(\sqrt{\theta_x\theta_p}
   -\sqrt{(1-\theta_x)(1-\theta_p)}\Big)^2,
\end{equation}
a strengthening of the Donoho--Stark principle~\cite{DonohoStark1989}. Throughout,
$f$ and $g$ denote the position and momentum densities.

\begin{figure}[t]
\includegraphics[width=0.99\linewidth]{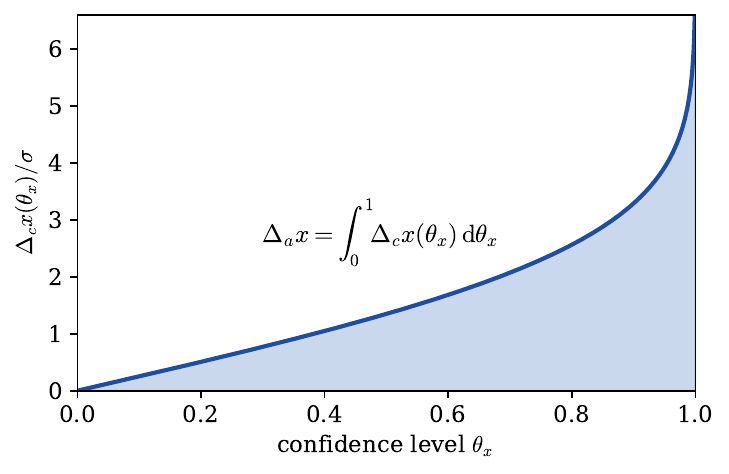}
\caption{The confidence width $\dx(\theta_x)$ of a Gaussian (in units of its standard
deviation $\sigma$) as a function of the
confidence level. The average confidence width $\Da x$ is the shaded area under the
curve. The integrable divergence as $\theta_x\to1$ reflects the Gaussian tails.}
\label{fig:conf}
\end{figure}

\section{The average confidence width}\label{sec:acw}

\begin{definition}[Average confidence width (ACW)]\label{def:spread}
The average confidence widths of position and momentum are
\begin{equation}\label{eq:defspread}
\begin{split}
  \Da x=\int_0^1\dx(\theta_x)\dd\theta_x, \\
  \Da p=\int_0^1\dpp(\theta_p)\dd\theta_p .
\end{split}
\end{equation}
For an observable with discrete spectrum the ACW is defined in the same way with the
counting measure in place of the Lebesgue measure.
\end{definition}

Geometrically, $\theta\mapsto\dx(\theta)$ is a (quantile-like) concentration curve and
$\Da x$ is the area beneath it (Fig.~\ref{fig:conf}), exactly as the ordinary mean is
the area under a quantile function. The width admits a closed form. Let $f^\ast$ denote
the decreasing rearrangement of $f$ (the non-increasing function on $[0,\infty)$ equimeasurable with $f$) and $\Theta(s)=\int_0^s f^\ast(u)\dd u$.

\begin{proposition}[Rearrangement identity and the MAD form]\label{prop:rearr}
For any state,
\begin{equation}\label{eq:rearr}
\begin{split}
  \Da x&=\int_0^\infty s\,f^\ast(s)\dd s=\int_0^\infty\bigl(1-\Theta(s)\bigr)\dd s \\
       &=2\int_\R|x|\,f^\sharp(x)\dd x,
\end{split}
\end{equation}
where $f^\sharp(x)=f^\ast(2|x|)$ is the symmetric decreasing rearrangement. Hence
$\Da x$ depends on the state only through $f^\ast$ (in particular it is invariant under
translations and under any measure-preserving rearrangement of $f$). Moreover
\begin{equation}\label{eq:variational}
  \Da x\ =\ \min_{V}\int_\R V(x)\,f(x)\dd x ,
\end{equation}
where the minimum runs over all measurable $V:\R\to[0,\infty)$ that are equimeasurable with $2|x|$,
i.e.\ $\mu\{V<t\}=t$ for all $t>0$; a minimizer is the \emph{rank potential}
$V_f(x)=\mu\{y:f(y)>f(x)\}$ (with ties resolved in a measure-preserving way). Consequently $\Da x$ is a
\emph{concave} functional of the density, and for every $x_0\in\R$
\begin{equation}\label{eq:madineq}
  \Da x\ \le\ 2\ev{|x-x_0|},
\end{equation}
with equality if and only if $f$ is symmetric and non-increasing about $x_0$; for such
\emph{symmetric unimodal} densities $\Da x=2\ev{|x-x_0|}$, twice the mean absolute deviation (MAD).
\end{proposition}

The proof (bathtub principle plus Fubini) is in Appendix~\ref{app:props}. Two
consequences fix the form of any uncertainty relation. Under the dilation
$\psi(x)\mapsto\lambda^{1/2}\psi(\lambda x)$ one has $\Da x\mapsto\lambda^{-1}\Da x$ and
$\Da p\mapsto\lambda\,\Da p$, so the product $\Da x\,\Da p$ is \emph{dilation
invariant} and any relation must take the scale-free form
\begin{equation}\label{eq:form}
  \Da x\,\Da p\ \ge\ c\,\hb,\qquad c>0 .
\end{equation}
We denote by $c^\ast=\inf_\rho \Da x\,\Da p/\hb$ the optimal constant. For a Gaussian,
$\Da x\,\Da p=(4/\pi)\hb\approx1.273\,\hb$; remarkably, the same value is obtained for the
two-sided exponential amplitude $\psi\propto e^{-|x|/\lambda}$ (Appendix~\ref{app:props}).
The remaining task is to determine $c^\ast$.

\section{The uncertainty relation}\label{sec:relation}

\subsection{Two rigorous lower bounds}

Because the two integrals in \eqref{eq:defspread} are over independent variables,
\begin{equation}\label{eq:productdouble}
  \Da x\,\Da p=\int_0^1\!\!\int_0^1\dx(\theta_x)\,\dpp(\theta_p)\,\dd\theta_x\dd\theta_p .
\end{equation}
On the region $\theta_x+\theta_p>1$ we bound the integrand by \eqref{eq:thm2};
on $\theta_x+\theta_p\le1$ we use only $\dx\,\dpp\ge0$ (Theorem~1 of Ref.~\cite{LinConfidence2026}
forbids any positive universal bound there). Hence

\begin{proposition}\label{cor:selfcontained}
$\displaystyle \Da x\,\Da p\ \ge\ 2\pi\hb\, I$, where
\begin{equation}
\begin{split}
I&= \iint_{\theta_x+\theta_p>1}\Big(\sqrt{\theta_x\theta_p}-\sqrt{(1-\theta_x)(1-\theta_p)}\Big)^2\dd\theta_x\dd\theta_p  \\
    &=\frac14-\frac{\pi^2}{64}\approx0.09579 .
\end{split}
\end{equation}
Equivalently $\Da x\,\Da p\ge\big(\tfrac{\pi}{2}-\tfrac{\pi^3}{32}\big)\hb\approx0.6019\,\hb$.
\end{proposition}

The proof is in Appendix~\ref{app:proofselfcontained}. This bound uses nothing beyond the results of \cite{LinConfidence2026}; since the latter are not
sharp, the constant $0.602$ is far from optimal. The next bound is much tighter.

\emph{The entropic bound.} Let $h(x)=-\int_\R f\ln f\dd x$ be the differential
entropy and recall the BBM
relation~\cite{BBM1975,Beckner1975,Hirschman1957} $h(x)+h(p)\ge\ln(\pi e\hb)$,
saturated by Gaussians. The average confidence width is tied to the entropy through a
one-dimensional extremal fact (the exponential maximizes entropy at fixed mean on a
half-line), giving the following result; the proof is in Appendix~\ref{app:relation}.

\begin{proposition}[Entropic bound]\label{thm:entropic}
For every state,
\begin{equation}
  \Da x\ \ge\ e^{\,h(x)-1},
\end{equation}
with equality if and only if $f^\ast$ is exponential; hence
\begin{equation}\label{eq:entropicspread}
  \Da x\,\Da p\ \ge\ e^{\,h(x)+h(p)-2}\ \ge\ \frac{\pi}{e}\,\hb\ \approx\ 1.1557\,\hb .
\end{equation}
\end{proposition}

(The single-marginal bound $\Da x\ge e^{h(x)-1}$ is stated in fixed units; the additive
$\ln(\text{length})$ ambiguity of the differential entropy under rescaling cancels in the
dilation-invariant product $\Da x\,\Da p$, exactly as in the BBM relation.) The two
inequalities in \eqref{eq:entropicspread} are saturated by different states: the first requires both
rearranged densities $f^\ast$, $g^\ast$ to be exponential, while Gaussians, which saturate BBM, have
half-Gaussian rearrangements and lose the factor $(4/\pi)/(\pi/e)\approx1.10$ in the first
step. Conversely, comb-like states (Sec.~\ref{sec:numerics}) have exactly exponential
rearrangements but a BBM deficit. The bound $\pi/e$ is therefore not expected to be
optimal, and indeed all our evidence points to a value of $c^\ast$ about $5\%$ larger. (The same mechanism limits the Hirschman-type entropic bound $\pi/4e$ for the product of mean absolute
deviations~\cite{Hirschman1957,BuschLahtiWerner2014}; what is new in \eqref{eq:entropicspread} is that the
entropy controls the \emph{rearranged} first moment, which can be much smaller than the mean absolute
deviation.) We note that replacing the Shannon entropy by R\'enyi
entropies $h_\alpha$, with the corresponding maximum-entropy laws at fixed mean and the
R\'enyi uncertainty relation of Ref.~\cite{BialynickiBirula2006}, yields a one-parameter
family of bounds that is maximal at the Shannon point $\alpha=1$; it does not improve on
$\pi/e$.

\subsection{The upper bound: reduction to $|x|+|p|$}

The MAD inequality \eqref{eq:madineq} gives
$\Da x\,\Da p\le4\ev{|x|}\ev{|p|}$ for \emph{every} state, with equality when both
densities are symmetric unimodal about the origin. Minimizing the right-hand side---a quantity governed by
a single Fourier-invariant operator---therefore bounds $c^\ast$ from above, and its
minimizer turns out to be symmetric unimodal, so the bound is attained.

\begin{theorem}[Reduction to $|x|+|p|$]\label{prop:hamiltonian}
Fix a reference length $\ell$ and the corresponding momentum $p_\ell=\hb/\ell$, and let
\begin{equation}\label{eq:Odef}
  \hat O\;=\;\frac{|x|}{\ell}+\frac{|p|}{p_\ell}
        \;=\;\frac{\sqrt{x^2}}{\ell}+\frac{\ell\,\sqrt{p^2}}{\hb}
\end{equation}
be the dimensionless self-adjoint operator built from $|x|=\sqrt{x^2}$ and
$|p|=\sqrt{p^2}$. In the canonical pair $X=x/\ell$,
$P=\ell p/\hb$ ($[X,P]=i$) it reads $\hat O=|X|+|P|$, manifestly invariant under the
Fourier transform $X\leftrightarrow P$. Then:

(a) $\hat O$ has purely discrete spectrum. Its lowest eigenvalue
$E_0=\min_{\norm\psi=1}\ev{\hat O}$ is independent of $\ell$ and nondegenerate, and
$\inf_\psi 4\ev{|x|}\ev{|p|}=E_0^2\,\hb$.

(b) The ground state $\psi_0$ can be chosen strictly positive (it has no nodes). It is even, equal
to its own Fourier transform, $\tilde\psi_0=\psi_0$ (for $\ell=1$, $\hb=1$), and strictly decreasing in $|x|$. Hence both
densities of $\psi_0$ are symmetric unimodal, $\Da x=2\ev{|x|}=E_0\ell$, $\Da p=2\ev{|p|}=E_0\hb/\ell$, and
\begin{equation}
  \Da x\,\Da p\,[\psi_0]=E_0^2\,\hb,\qquad\text{so that}\qquad c^\ast\le E_0^2 .
\end{equation}

(c) $\psi_0$ has algebraic tails: for $\ell=\hb=1$ there are constants $0<c_1\le c_2$ such that
$c_1(1+|x|)^{-3}\le\psi_0(x)\le c_2(1+|x|)^{-3}$ for all $x\in\R$.
\end{theorem}

The proof is given in Appendix~\ref{app:relation}; in brief, nodelessness and nondegeneracy follow from the
nonlocality of $|p|$, whose semigroup has a strictly positive (Poisson) kernel (a Perron--Frobenius argument), evenness and
self-duality from nondegeneracy, monotonicity from the fact that symmetric decreasing
rearrangement lowers both $\ev{|x|}$ and $\ev{|p|}$, and the tail bound (c) from general ground-state estimates for
nonlocal Schr\"odinger operators~\cite{KaletaLorinczi2015}. A formal asymptotic analysis of the eigenvalue equation
(Appendix~\ref{app:relation}), confirmed by the numerical solution, refines (c): for $\ell=\hb=1$,
\begin{equation}\label{eq:asympt}
\begin{split}
  \psi_0(x)&=\psi_0(0)+\pi^{-1}\psi_0(0)\,x^2\ln|x|+O(x^2)\quad(x\to0),\\
  \psi_0(x)&\simeq C|x|^{-3}\quad(|x|\to\infty),\qquad C=\frac{2\psi_0(0)}{\sqrt{2\pi}} .
\end{split}
\end{equation}
Numerically, $E_0=1.1040744(1)$, $E_0^2=1.2189803$, $\psi_0(0)^2=0.6659$, and $C=0.651$.
Two independent computations of $E_0$ (a lattice
discretization with Richardson extrapolation, and variational Rayleigh--Ritz upper bounds in the
Hermite basis, Appendix~\ref{app:num}) agree to eight digits, and $E_0^2/4=0.3047451$ agrees with the constant
$0.304745$ computed by Werner~\cite{Werner2004}; the numerical results of
Refs.~\cite{Werner2004,BuschLahtiWerner2014} already indicated an even, Fourier-invariant, non-Gaussian ground
state, and part (b) supplies the proof. Combining Proposition~\ref{thm:entropic} and
Theorem~\ref{prop:hamiltonian},
\begin{equation}\label{eq:chain}
  \underbrace{\tfrac{\pi}{e}}_{\approx1.1557}\hb
  \ \le\ c^\ast\hb\ \le\
  \underbrace{E_0^2}_{\approx1.2190}\hb
  \ <\
  \underbrace{\tfrac{4}{\pi}}_{\approx1.2732}\hb .
\end{equation}
The strict last inequality is the headline qualitative fact: \emph{the Gaussian is not
the minimizer}, in sharp contrast with the Kennard and entropic relations, both
saturated by Gaussians. The state $\psi_0$ is more sharply peaked than a Gaussian of the same variance (its
curvature at the origin diverges logarithmically), has lighter shoulders, and has heavy algebraic tails,
$|\psi_0|^2\simeq C^2|x|^{-6}$ (kurtosis $\approx5.6$); see Fig.~\ref{fig:profile}. These features are
inherited from the kinks of $|x|$ and $|p|$. Within the
one-parameter family $\psi\propto e^{-|x|^n/2}$ the product is minimized at $n\approx1.37$ with value
$\approx1.2250$, just above $E_0^2$; it equals $4/\pi$ both at $n=2$ and at $n=1$ (Fig.~\ref{fig:vsn}).

\begin{figure}[t]
  \centering
    \includegraphics[width=0.95\linewidth]{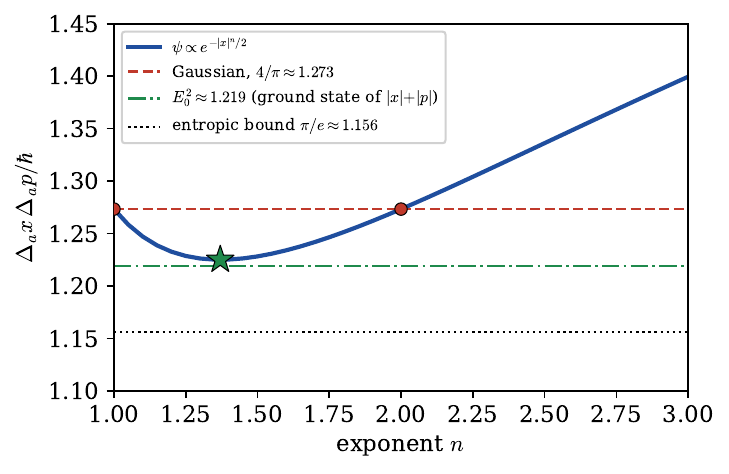}
    \caption{The product $\Da x\,\Da p$ for $\psi\propto e^{-|x|^n/2}$: the minimum (star) sits at
    $n\approx1.37$ slightly above $E_0^2$ (dash-dotted), strictly below the value $4/\pi$ (dashed)
    attained at $n=2$ and $n=1$ (circles). Dotted: the rigorous lower bound $\pi/e$.}
    \label{fig:vsn}
\end{figure}

\begin{figure}[t]
    \includegraphics[width=\linewidth]{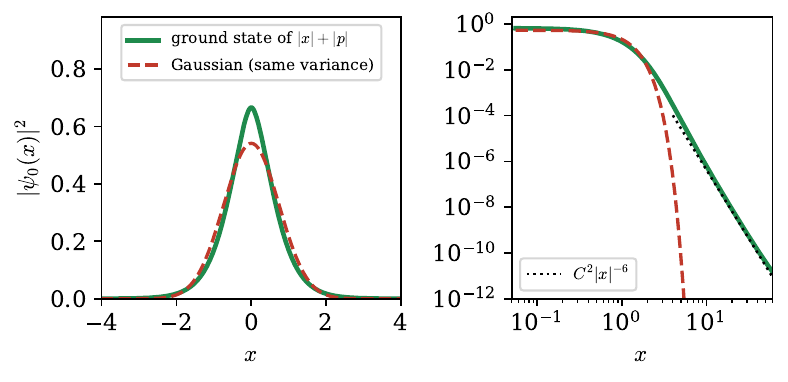}
    \caption{Position density of the ground state $\psi_0$ of $|x|+|p|$ (solid) versus a
    variance-matched Gaussian (dashed). Left: linear scale; $\psi_0$ is more sharply peaked, with lighter
    shoulders. Right: logarithmic scales, showing the algebraic tail $C^2|x|^{-6}$ (dotted) of
    Eq.~\eqref{eq:asympt}. By self-duality the momentum density is identical.}
    \label{fig:profile}
\end{figure}

\subsection{Sharp hybrid relations}\label{sec:hybrid}

The difficulty of proving a lower bound that matches $E_0^2$ is that $\Da x$ and $\Da p$ are
controlled by the \emph{individual} rearrangements $f^\sharp$ and $g^\sharp$, which are in general
not the densities of a common state. If only one of the two widths is replaced by the mean absolute
deviation, the problem can be solved completely.

\begin{theorem}[Hybrid relations]\label{thm:hybrid}
For every state (pure or mixed) and all $x_0,p_0\in\R$,
\begin{equation}\label{eq:hybrid}
\begin{split}
  \Da x\cdot 2\ev{|p-p_0|}\ &\ge\ E_0^2\,\hb ,\\
  2\ev{|x-x_0|}\cdot\Da p\ &\ge\ E_0^2\,\hb .
\end{split}
\end{equation}
Both relations are sharp; equality holds if and only if the state is pure and is a dilated, translated, and boosted copy
of $\psi_0$ (up to a phase) whose mean momentum is $p_0$ in the first relation, or whose center is $x_0$ in the second.
\end{theorem}

The proof (Appendix~\ref{app:relation}) rearranges the wave function in one representation only:
passing from $\psi$ to the symmetric decreasing rearrangement of $|\psi|$ leaves $\Da x$ unchanged
and does not increase the ``kinetic energy'' $\ev{|p|}$~\cite{LiebLoss2001}, after which
Theorem~\ref{prop:hamiltonian} applies. Since $\Da p=2\ev{|p-p_0|}$ whenever $g$ is symmetric
unimodal about $p_0$, we obtain a rigorous statement of the conjectured optimal relation on a large
class of states.

\begin{corollary}\label{cor:unimodal}
If at least one of the two densities $f$, $g$ is symmetric unimodal (about any center), then
$\Da x\,\Da p\ge E_0^2\,\hb$, irrespective of the shape of the other density.
\end{corollary}

In particular, any state violating $\Da x\,\Da p\ge E_0^2\hb$ would need \emph{both} densities to
be asymmetric or multimodal. Since $\Da x\le2\ev{|x-x_0|}$, relations \eqref{eq:hybrid} strengthen the
sharp uncertainty relation for mean absolute deviations of
Refs.~\cite{Werner2004,BuschLahtiWerner2014},
$\ev{|x-x_0|}\ev{|p-p_0|}\ge E_0^2\hb/4\approx0.3047\,\hb$ (the Gaussian gives $\hb/\pi\approx0.3183\,\hb$):
one of the two deviations may be replaced by the smaller, rearrangement-invariant quantity $\Da/2$.

\subsection{Second variation at $\psi_0$}\label{sec:local}

The functional $J[\psi]=\Da x\,\Da p$ is invariant under translations, boosts, dilations, and phase
changes. We now show that, apart from these symmetries, every smooth perturbation of $\psi_0$
increases $J$ at second order. (This is a statement about each fixed smooth direction, and its proof uses the
asymptotic forms \eqref{eq:asympt}; because $J$ is a
minimum over rearrangements and hence not twice differentiable in any norm, it does not by itself exclude
lower values at nearby rough states, for which we rely on Corollary~\ref{cor:unimodal} and on the numerical
evidence below.) Set $\ell=\hb=1$, $H=|x|+|p|$, and let $\eta$ be a Schwartz function
orthogonal to $\psi_0$, with even part $\eta_e$ and odd part $\eta_o=\eta_1+i\eta_2$ ($\eta_{1,2}$ real);
let $-i\chi_2$ be the Fourier transform of $\eta_2$ ($\chi_2$ is real and odd).

\begin{proposition}[Second variation]\label{prop:local}
For $\psi_\eps=(\psi_0+\eps\eta)/\norm{\psi_0+\eps\eta}$,
\begin{equation}\label{eq:secondvar}
\begin{split}
  J[\psi_\eps]=E_0^2+\eps^2\Big\{&Q_e[\eta_e]+2E_0\big(\ev{\eta_1|L|\eta_1}\\
  &+\ev{\chi_2|L|\chi_2}\big)\Big\}+o(\eps^2),
\end{split}
\end{equation}
where
\begin{equation}
  Q_e[\eta_e]=2E_0\ev{\eta_e|(H-E_0)|\eta_e}-16\big(\mathrm{Re}\ev{\psi_0||x||\eta_e}\big)^2\ \ge0
\end{equation}
is the second variation of $4\ev{|x|}\ev{|p|}$, which vanishes only along the dilation mode, and
\begin{equation}\label{eq:Lop}
  L=|p|+|x|-E_0-\frac{\psi_0(x)}{|\psi_0'(x)|}
\end{equation}
acts on odd functions. The operator $L$ is nonnegative on the odd sector, and its only zero mode is
$\psi_0'$, the generator of translations (for $\chi_2$: of boosts). Hence the second variation of $J$
at $\psi_0$ is nonnegative and vanishes only along symmetry directions.
\end{proposition}

The term $-\psi_0/|\psi_0'|$ is the second-order \emph{gain} from rearrangement: an odd perturbation
displaces the level sets of the density without changing their lengths to first order, and for a
unimodal density one has the exact identity
$2\ev{|x-c|}-\Da x=2\int_0^{\max f}[m(t)-c]^2\dd t$, where $m(t)$ is the midpoint of the level set
$\{f>t\}$ and $c$ the mode. Nonnegativity of $L$ follows from a Perron--Frobenius argument on the half line: $\psi_0'$ is an
eigenfunction of $L$ with eigenvalue $0$ that does not change sign on $(0,\infty)$, hence it is the
ground state of $L$ in the odd sector (Appendix~\ref{app:relation}). Numerically the spectrum of $L$
on odd functions is $\{0,\,1.11,\,1.92,\dots\}$, so the zero mode is isolated by a gap of order one.

\subsection{Numerical evidence and the conjecture}\label{sec:numerics}

The variational formula \eqref{eq:variational} turns the minimization of $J$ into a \emph{double}
minimization,
\begin{equation}\label{eq:double}
  c^\ast\hb=\inf_{\psi}\ \min_{V,\tilde V}\ \ev{V(x)}_\psi\,\ev{\tilde V(p)}_\psi ,
\end{equation}
over states and over pairs of potentials equimeasurable with $2|x|$ and $2|p|$. Both partial
minimizations can be carried out exactly: for fixed $\psi$ the optimal potentials are the rank
potentials $V_f$, $\tilde V_g$; for fixed $(V,\tilde V)$ the optimal state is the ground state of
$\lambda V(x)+\tilde V(p)/\lambda$ with $\lambda$ chosen to minimize the ground-state energy, since
$\inf_\psi\ev{V}\ev{\tilde V}=\min_\lambda[E_{\rm gs}(\lambda V+\tilde V/\lambda)/2]^2$. Alternating the two steps
gives a self-consistent-field (SCF) iteration under which $J$ \emph{decreases monotonically}. In this
language the conjecture $c^\ast=E_0^2$ states that, among all Hamiltonians $V(x)+\tilde V(p)$ with $V$ and
$\tilde V$ equimeasurable with $|x|$ and $|p|$, the ``diamond'' $|x|+|p|$ has the lowest ground-state energy.
(All these Hamiltonians share the same phase-space volume of $\{V+\tilde V<E\}$, hence the same Weyl
asymptotics; they differ in their low-lying spectrum.) For fixed $\tilde V=|p|$ the optimal $V$ is indeed $|x|$, by
Theorem~\ref{thm:hybrid}, and vice versa; the open question is the simultaneous rearrangement.

We ran the SCF iteration in three kinds of variational spaces (Appendix~\ref{app:num}): (i) the span of the first $K$ Hermite
functions with complex coefficients and both parities, $K\le128$; (ii) the span of $K$
coherent states (Gaussians) placed on a lattice in phase space, $K\le121$, the natural
overcomplete family in which any wave function can be expanded; and (iii) the full discretized
Hilbert space $\C^N$ with the discrete Fourier transform, $N\le512$, i.e.\ the Fourier pair of
mutually unbiased bases on $\mathbb{Z}_N$, which imposes no smoothness at all. Starting points
included random coefficient vectors, random superpositions of coherent states, Fock states,
Schr\"odinger-cat states, comb states, box and exponential profiles, chirped and cubic-phase
Gaussians, and white noise; in space (ii) an independent quasi-Newton minimization (L-BFGS with the
analytic gradient) was run from the same starting points. The outcome is summarized in
Table~\ref{tab:numerics} and Fig.~\ref{fig:conv}. In the spaces (i) and (ii) \emph{every} SCF run converged to
$\psi_0$ up to translation, boost, and dilation (fidelity above $0.9999$ for $K\ge32$ and
$1-10^{-5}$ or better for the largest spaces). The quasi-Newton runs reached the same state, except for a
few (about one in ten) that ended in a second type of local minimum of the restricted problem: comb-like states
with four to six teeth in each density and product $1.572$--$1.575$, just above the comb limit $\pi/2$ derived below. No run in any space produced a product below
$E_0^2$, and the minimum over each
space approaches $E_0^2$ from above as the space grows, algebraically in $K$ as expected for a function with
$|x|^{-3}$ tails and an $x^2\ln|x|$ singularity. On the full lattice the runs end in lattice copies of
$\psi_0$ (or, on the coarsest lattice, in lattice-scale comb states that have no continuum counterpart), and all
values carry the expected $O(1/N)$ discretization shift (Appendix~\ref{app:num}).

\begin{table}[tb]
\caption{Global minimization of $J=\Da x\,\Da p$ (units $\hb=1$) by the SCF iteration from random and
structured starting points; $E_0^2=1.2189803$ and $K$ is the dimension of the variational space ($K=N$ for the
full lattice). For coherent states the minimum also includes the quasi-Newton
runs. ``Fidelity'' is the overlap $|\ev{\psi_0|\psi}|^2$ of the
final state with $\psi_0$, maximized over translations, boosts, and dilations (smallest value over all
SCF runs); for the full lattice see Appendix~\ref{app:num}.}
\label{tab:numerics}
\begin{ruledtabular}
\begin{tabular}{lrrcc}
variational space & $K$ & runs & $\min J-E_0^2$ & fidelity \\
\colrule
Hermite functions & 16 & 12 & $1.5\times10^{-3}$ & 0.9953 \\
 & 32 & 16 & $4.0\times10^{-4}$ & 0.99993 \\
 & 64 & 16 & $1.6\times10^{-4}$ & 0.99996 \\
 & 128 & 12 & $6.3\times10^{-5}$ & 0.999996 \\
coherent states & 21 & 8 & $9.9\times10^{-4}$ & 0.9946 \\
 & 69 & 8 & $1.6\times10^{-4}$ & 0.99993 \\
 & 121 & 8 & $6.8\times10^{-5}$ & 0.99999 \\
full lattice $\mathbb{Z}_K$ & 128 & 14 & $1.7\times10^{-3}$ & --- \\
 & 256 & 14 & $1.5\times10^{-3}$ & --- \\
 & 512 & 14 & $7.5\times10^{-4}$ & --- \\
\end{tabular}
\end{ruledtabular}
\end{table}

\begin{figure}[t]
    \includegraphics[width=0.95\linewidth]{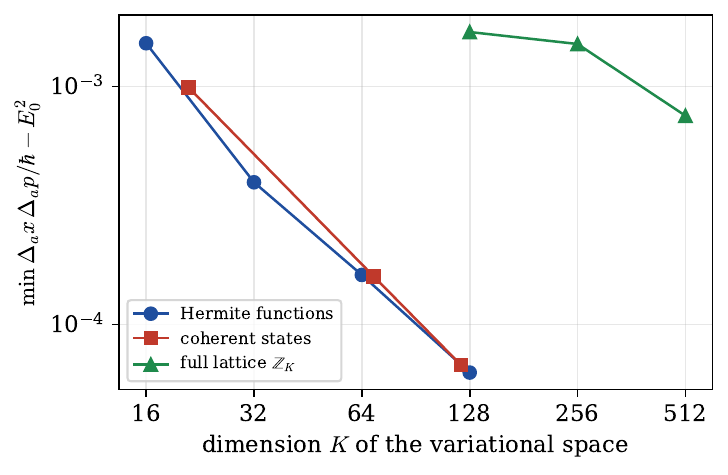}
    \caption{Smallest product $\Da x\,\Da p$ found in each variational space minus $E_0^2$, versus the
    size of the space. All minima lie above $E_0^2$ and approach it as the space grows.}
    \label{fig:conv}
\end{figure}

\emph{Comb states.} The states most likely to defeat a rearrangement-invariant relation are
comb-like superpositions, which underlie the joint localizability of
Ref.~\cite{LinConfidence2026} and are familiar as Gottesman--Kitaev--Preskill
states~\cite{GKP2001}: $\psi(x)\propto\sum_n e^{-(na)^2/2\Sigma^2}e^{-(x-na)^2/2\delta^2}$ with
$\delta\ll a\ll\Sigma$. Their densities consist of many narrow teeth in \emph{both} representations. In
this limit each density factorizes into an envelope and a tooth profile, and its rearrangement is
that of a two-dimensional product density; for Gaussian envelope and teeth one finds
$\Da x=\pi\Sigma\delta/a$, $\Da p=\hb a/(2\Sigma\delta)$ and
\begin{equation}
  \Da x\,\Da p\ \to\ \frac{\pi}{2}\,\hb\approx1.571\,\hb ,
\end{equation}
independently of $a,\Sigma,\delta$, and in fact $\pi/2$ is the \emph{smallest} value attainable in the
comb limit for any envelope and tooth shapes, by the two-dimensional result of
Sec.~\ref{sec:dimension}. Remarkably, in this limit the rearranged densities of the Gaussian comb become
exactly exponential, so they saturate the first inequality of \eqref{eq:entropicspread}; they lose to
$\psi_0$ because of their large entropy sum.

We therefore conjecture:

\begin{conjecture}\label{conj:opt}
The optimal constant in \eqref{eq:form} is $c^\ast=E_0^2=1.2189803\ldots$, attained uniquely (up to
dilation, translation, boost, and phase) by the ground state of $|x|+|p|$.
\end{conjecture}

The conjectured relation holds rigorously on the class of Corollary~\ref{cor:unimodal}, and $\psi_0$ is
stable at second order (Proposition~\ref{prop:local}). It would follow in general from the inequality
$\Da p[\psi^\sharp]\le\Da p[\psi]$, where $\psi^\sharp$ is the symmetric decreasing rearrangement of
$|\psi|$, i.e.\ from the statement that symmetrization in position space does not spread the
momentum density in the sense of the ACW; we found no counterexample to this inequality in
$3\times10^3$ random trials, but we have no proof. Problems of this type---simultaneous control
of a function and its Fourier transform under rearrangement---are notoriously hard; the conjecture
of Donoho and Stark that intervals maximize time--frequency concentration~\cite{DonohoStark1989,DonohoStark1993} is
a well-known relative that has been proved only for small time--bandwidth products.

\subsection{Mixed states}\label{sec:mixed}

All definitions depend on the state only through its densities $f(x)=\bra x\rho\ket x$ and
$g(p)=\bra p\rho\ket p$, so they apply verbatim to mixed states. Nothing is lost, and something is
gained:

\begin{proposition}[Mixed states]\label{prop:mixed}
(a) Propositions~\ref{prop:rearr}--\ref{thm:entropic}, the confidence bound \eqref{eq:thm2} and
Theorem~\ref{thm:hybrid} hold for all density operators.
(b) $\Da x$ and $\Da p$ are concave functionals of $\rho$. Consequently, for $\rho=\sum_iq_i\rho_i$,
\begin{equation}\label{eq:mixedCS}
  \Da x[\rho]\,\Da p[\rho]\ \ge\ \Big(\sum_iq_i\sqrt{\Da x[\rho_i]\,\Da p[\rho_i]}\Big)^2 ,
\end{equation}
so the optimal constant over all density operators equals the pure-state constant $c^\ast$, and
the product for a mixture is never smaller than the smallest product among its components.
(c) For a state of von Neumann entropy $S(\rho)=-\Tr\rho\ln\rho$,
\begin{equation}\label{eq:mixedS}
  \Da x\,\Da p\ \ge\ \max\Big\{\frac{\pi}{e},\ \frac{2\pi}{e^2}\,e^{S(\rho)}\Big\}\,\hb .
\end{equation}
\end{proposition}

Part (c) uses the entropic relation of Frank and Lieb~\cite{FrankLieb2012},
$h(x)+h(p)\ge\ln(2\pi\hb)+S(\rho)$, which is stronger than BBM when $S(\rho)>\ln(e/2)$. For a thermal
state of the harmonic oscillator with mean occupation $\bar n$ both densities are Gaussian,
$\Da x\,\Da p=(8/\pi)(\bar n+\tfrac12)\hb$, while \eqref{eq:mixedS} gives asymptotically
$(2\pi/e)\,\bar n\,\hb$; the ratio of bound to actual value tends to $\pi^2/4e\approx0.91$, so
\eqref{eq:mixedS} captures the correct growth with the mixedness. The proof of
Proposition~\ref{prop:mixed} is in Appendix~\ref{app:relation}.

\subsection{Why the Gaussian fails: the role of dimension}\label{sec:dimension}

For a particle in $\R^d$ define the confidence width as the smallest $d$-dimensional volume carrying
probability $\theta$, and $\Da^{(d)}x$ as its average over $\theta$. Proposition~\ref{prop:rearr}
holds with $s$ the volume, the mean--entropy inequality is unchanged, and BBM reads
$h(x)+h(p)\ge d\ln(\pi e\hb)$, so that
\begin{equation}\label{eq:ddim}
  \Da^{(d)}x\,\Da^{(d)}p\ \ge\ \frac{(\pi e\hb)^d}{e^2}.
\end{equation}
For a radially symmetric decreasing density $\Da^{(d)}x=\omega_d\ev{|x|^d}$ with $\omega_d$ the volume of the unit
ball, and $\Da^{(d)}x\le\omega_d\ev{|x|^d}$ in general, so the role of $|x|+|p|$ is played by the Fourier-invariant operator $|x|^d+|p|^d$, and
$c^\ast_d\le\omega_d^2(E_0^{(d)}/2)^2$ with $E_0^{(d)}$ its ground-state energy. (For $d\ge3$ we do not know whether
this bound is attained, since $\ev{|p|^d}$ is then not known to decrease under rearrangement.) For $d=2$ this operator
is the isotropic harmonic oscillator, $E_0^{(2)}=2$, its ground state is the Gaussian, and the upper
bound $\pi^2$ coincides with the entropic lower bound \eqref{eq:ddim}:

\begin{proposition}\label{prop:2d}
In two dimensions $\Da^{(2)}x\,\Da^{(2)}p\ge\pi^2\hb^2$, and the constant is optimal; equality holds
for Gaussian wave functions $\psi\propto e^{-x\cdot Ax/2}$ with real positive-definite $A$ (and their translates and boosts).
\end{proposition}

The reason is geometric: the super-level sets $\{f>t\}$ of a two-dimensional Gaussian are ellipses whose area is
proportional to $\ln(f_{\max}/t)$, so its rearranged density is \emph{exactly exponential} and both
inequalities in the chain leading to \eqref{eq:ddim} are saturated simultaneously. Only for $d=2$ is
the volume a quadratic function of the radius. For $d=1$ the rearranged Gaussian is a half-Gaussian and the
optimizer is the ground state of $|x|+|p|$; for $d=3$ the Gaussian gives $256\pi/9\approx89.4$ against
the entropic bound $\pi^3e\approx84.3$ (in units of $\hb^3$). In this sense the non-Gaussian optimizer of the
one-dimensional relation is a consequence of measuring localization by \emph{length} rather than by a
quadratic form.

\section{Complementary observables and wave--particle duality}\label{sec:compl}

An uncertainty relation of the form \eqref{eq:form} is a statement about a pair of \emph{complementary}
observables, i.e.\ observables whose eigenbases are mutually unbiased. Position and momentum are one
such pair; wave--particle duality in an interferometer concerns another. In this section we make the
connection precise and delimit it.

\subsection{Wave--particle duality as the upper-bound face of position--momentum uncertainty}

The average confidence width measures localization: a small $\Da x$ means that the
probability is concentrated on a set of small total length, i.e.\ the state is
\emph{particle-like}; a small $\Da p$ means a concentrated momentum distribution, i.e.\
\emph{wave-like}. It is therefore natural to define the
\begin{equation}\label{eq:PW}
\begin{split}
  \text{particle character}\quad \mathcal{C}_{\rm part}[\psi]&=\frac1{\Da x}, \\
  \text{wave character}\quad \mathcal{C}_{\rm wave}[\psi]&=\frac1{\Da p}.
\end{split}
\end{equation}
The uncertainty relation \eqref{eq:form} with the optimal constant $c^\ast$ is, term for term, the same
statement as
\begin{equation}\label{eq:duality}
 \mathcal{C}_{\rm part}[\psi]\,\mathcal{C}_{\rm wave}[\psi]\ \le\ \frac{1}{c^\ast\hb} .
\end{equation}
No state can be simultaneously strongly particle-like and strongly wave-like; the
\emph{product} of the two characters is capped. This is the precise sense in which
\emph{uncertainty and duality are two faces of one inequality}: uncertainty read as a
\emph{lower} bound on a product of spreads, duality as an \emph{upper} bound on a
product of localizations. The minimizers of the uncertainty product are exactly the
maximizers of the duality product---the most balanced particle/wave compromise---and, by
Theorem~\ref{prop:hamiltonian}, Gaussian states are not among them. Reading $1/\Da p$ as the wave
character presupposes that the which-way information is carried by the position of the particle itself,
as in the double slit (Sec.~\ref{sec:slit}); the scope of this reading is delimited in
Sec.~\ref{sec:scope}.

\subsection{Other complementary pairs for a continuous variable}

Position and momentum are far from the only pair. The rotated quadratures
$Q_\vartheta=X\cos\vartheta+P\sin\vartheta$ (with $X=x/\ell$, $P=\ell p/\hb$) have generalized eigenfunctions
\begin{equation}
  \langle X|q\rangle_\vartheta=\frac{1}{\sqrt{2\pi|\sin\vartheta|}}\,
  \exp\!\Big[\frac{i}{\sin\vartheta}\Big(qX-\frac{X^2+q^2}{2}\cos\vartheta\Big)\Big],
\end{equation}
linear chirps of constant modulus, and
$|{}_\vartheta\langle q|q'\rangle_{\vartheta'}|^2=1/(2\pi|\sin(\vartheta-\vartheta')|)$ is independent of $q,q'$:
the eigenbases of $Q_\vartheta$, $0\le\vartheta<\pi$, are pairwise mutually
unbiased~\cite{WeigertWilkinson2008}, and each one (for $\vartheta\neq0,\pi/2$) is complementary to \emph{both}
position and momentum. Physically $Q_\vartheta$ is measured by homodyne detection with local-oscillator
phase $\vartheta$, by a position measurement after free evolution (time of flight) or after harmonic evolution
through the phase $\vartheta$, or, in
optics, in a Fresnel plane between the near field ($x$) and the far field ($p$). Since
$Q_\vartheta=\sin\vartheta\,U PU^\dagger$ with the chirp $U=e^{-iX^2\cot\vartheta/2}$, which commutes with
$X$, the pair $(X,Q_\vartheta)$ has exactly the same set of attainable pairs of densities as $(X,P)$ up to the
scale factor $\sin\vartheta$; together with rotation covariance this gives
\begin{equation}
  \Da q_\vartheta\,\Da q_{\vartheta'}\ \ge\ c^\ast\,|\sin(\vartheta-\vartheta')| ,
\end{equation}
with the \emph{same} optimal constant $c^\ast$ as in \eqref{eq:form}. For three quadratures at mutual angles $\pi/3$
(a symmetric triple~\cite{WeigertWilkinson2008,KechrimparisWeigert2014}), multiplying the three pairwise relations gives
$\Da q_0\,\Da q_{\pi/3}\,\Da q_{2\pi/3}\ge(\sqrt3c^\ast/2)^{3/2}$.

\subsection{The path qubit of a two-way interferometer}

In a Mach--Zehnder interferometer the relevant degree of freedom is the path qubit, with
$\sigma_z=\ket0\bra0-\ket1\bra1$ the which-way observable. A phase shifter $e^{i\phi\sigma_z/2}$
followed by a balanced beam splitter and detection of the output ports measures
$\sigma_\phi=\sigma_x\cos\phi+\sigma_y\sin\phi$, so the port probabilities are
\begin{equation}\label{eq:fringe}
  p_\pm(\phi)=\tfrac12\big[1\pm(r_x\cos\phi+r_y\sin\phi)\big]=\tfrac12\big[1\pm\mathcal V\cos(\phi-\phi_0)\big],
\end{equation}
where $\vec r=\Tr(\rho\vec\sigma)$ is the Bloch vector. The predictability and the fringe visibility
are
\begin{equation}
  \mathcal P=|r_z|,\qquad \mathcal V=\sqrt{r_x^2+r_y^2}=\max_\phi|\ev{\sigma_\phi}| .
\end{equation}
Thus the interferometer involves all \emph{three} mutually unbiased bases of the qubit: $\sigma_z$ is the
which-way observable, and $\sigma_x$, $\sigma_y$ are the in-phase and quadrature components of the fringe
pattern ($r_x=\mathcal V\cos\phi_0$, $r_y=\mathcal V\sin\phi_0$); a fixed $\pi/2$ phase shift in one arm converts a
measurement of $\sigma_x$ into one of $\sigma_y$. The visibility is not the expectation of a single one of them but the
length of the equatorial component, i.e.\ the largest $|\ev{\sigma_\phi}|$ over the one-parameter
family of observables complementary to the path. The duality relation
$\mathcal P^2+\mathcal V^2\le1$~\cite{WoottersZurek1979,GreenbergerYasin1988,JaegerShimonyVaidman1995,Englert1996} is then nothing but
$r_x^2+r_y^2+r_z^2\le1$, a relation among three mutually complementary observables. This structure is
known: Brukner and Zeilinger~\cite{BruknerZeilinger1999,BruknerZeilinger2002} identified the path and the
output ports at phases $\phi$ and $\phi+\pi/2$ as three mutually complementary propositions whose information
contents add up to $2\Tr\rho^2-1$, and formulations in terms of uncertainty relations for
$\sigma_x,\sigma_y,\sigma_z$ were given in
Refs.~\cite{BjorkEtAl1999,DurrRempe2000,BuschShilladay2006,BosykEtAl2013}; the entropic versions use the path
basis and the optimal equatorial basis~\cite{ColesKaniewskiWehner2014}.

The ACW fits this structure exactly. For a two-outcome distribution the discrete width is
$W=1\cdot p_{\max}+2(1-p_{\max})=2-p_{\max}$, so for the three Pauli observables
\begin{equation}\label{eq:qubitW}
  W_k=\frac{3-|r_k|}{2},\qquad k=x,y,z,
\end{equation}
and therefore
\begin{equation}\label{eq:qubitPV}
\begin{split}
  \mathcal P&=3-2W_z,\qquad \mathcal V=3-2\min_\phi W_\phi ,\\
  &\sum_{k=x,y,z}(3-2W_k)^2=2\Tr\rho^2-1\ \le\ 1 .
\end{split}
\end{equation}
For the path qubit, the conventional particle and wave measures \emph{are} (affine functions of) the
average confidence widths of the path observable and of the optimal complementary observable, and the
duality relation is an uncertainty relation for ACWs. In product form,
$W_zW_\phi\ge(3-1/\sqrt2)^2/4\approx1.314$ for every pair and $W_xW_yW_z\ge[(3-1/\sqrt3)/2]^3\approx1.777$,
attained at the balanced states $\mathcal P=\mathcal V=1/\sqrt2$ and $|r_x|=|r_y|=|r_z|=1/\sqrt3$ respectively; these
are the qubit counterparts of \eqref{eq:form}, with the role of the non-Gaussian optimizer
played by the balanced state rather than by a basis state.

\subsection{The double slit}\label{sec:slit}

In a double-slit experiment the two complementary observables \emph{are} position and momentum: the
which-slit information resides in the transverse position density $f$ just behind the slits, and the
far-field interference pattern is the transverse momentum density $g$. Consider two identical,
non-overlapping slits with profile $\varphi$ centered at $\pm a$, weights $w_{1,2}$ and mutual
coherence $\gamma$, so that $\mathcal P=|w_1-w_2|$ and $\mathcal V=2|\gamma|\sqrt{w_1w_2}$. Then
$f=w_1|\varphi(x-a)|^2+w_2|\varphi(x+a)|^2$ and $g(p)=|\tilde\varphi(p)|^2[1+\mathcal V\cos(2ap/\hb-\phi_0)]$, and one finds
(Appendix~\ref{app:slit}):

\begin{proposition}\label{prop:slit}
$\Da x$ is a strictly decreasing function of $\mathcal P$, from $2\Da x^{(1)}$ at $\mathcal P=0$ to the single-slit value
$\Da x^{(1)}$ at $\mathcal P=1$. In the many-fringe regime $\Da p$ is a strictly decreasing function of $\mathcal V$, from the
single-slit value $\Da p^{(1)}$ at $\mathcal V=0$ to $R_1\Da p^{(1)}$ at $\mathcal V=1$, with $\tfrac12<R_1<1$. For a uniform slit
profile $\Da x=(2-\mathcal P)\,\Da x^{(1)}$, and for a uniform envelope $\Da p=(1-4\mathcal V/\pi^2)\,\Da p^{(1)}$.
\end{proposition}

(The two closed forms are idealizations that illustrate one marginal each: a uniform slit profile has an
envelope with $\Da p^{(1)}=\infty$, and a uniform envelope requires slit amplitudes with slowly decaying tails.)

So in the double-slit setting $1/\Da x$ and $1/\Da p$ are monotone in the conventional
predictability and visibility (Fig.~\ref{fig:slit}), and the relation \eqref{eq:form} constrains how
large both can be. This is the sense in which the position--momentum relation \eqref{eq:form} can be
read as a wave--particle duality relation.

\begin{figure}[t]
    \includegraphics[width=\linewidth]{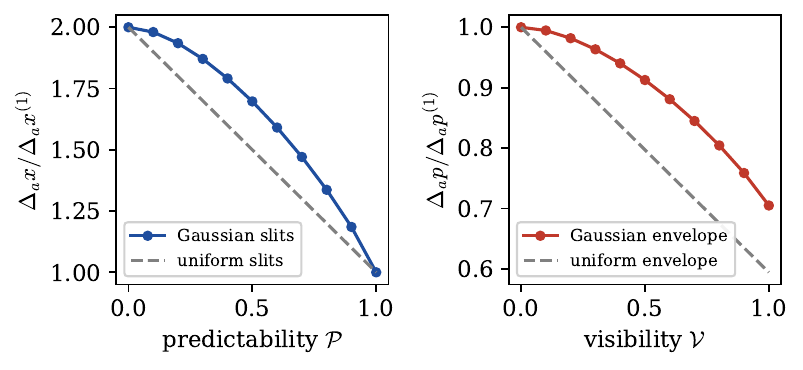}
    \caption{Double slit. Left: $\Da x$ in units of its single-slit value versus the which-slit
    predictability $\mathcal P$. Right: $\Da p$ in units of its single-slit value versus the fringe visibility $\mathcal V$
    (many-fringe regime). Symbols: Gaussian slit profile; dashed: uniform profile and uniform envelope,
    Proposition~\ref{prop:slit}.}
    \label{fig:slit}
\end{figure}

\subsection{Scope of the duality reading}\label{sec:scope}

The examples show what can and what cannot be claimed. Wave--particle duality is the uncertainty
relation between the which-way observable of a given interferometer and the observables complementary to
it \emph{within the degree of freedom that carries the path information}. For the double slit this pair is
$(x,p)$, and the ACW relation \eqref{eq:form} applies directly. For a Mach--Zehnder interferometer it is
$(\sigma_z,\sigma_\phi)$, and the relevant ACW relation is \eqref{eq:qubitPV}; the momentum distribution of
the wave packets traveling along the arms is a different degree of freedom, and its width $\Da p$ bears no
relation to the visibility. We therefore do \emph{not} propose $1/\Da p$ as a device-independent measure of
``wave character'' for all cases; rather, uncertainty relations for complementary observables are the more general
notion, and duality relations are their instances in interferometric settings, in line with
Refs.~\cite{DurrRempe2000,BuschShilladay2006,ColesKaniewskiWehner2014,SpegelLexne2024}. A systematic
treatment of the ACW for arbitrary numbers of mutually unbiased bases in finite dimension will be given
elsewhere.

\section{Connections}\label{sec:lit}

The average confidence width sits at the
intersection of several traditions and is complementary to each. The
\emph{variance} ($L^2$) relations~\cite{Heisenberg1927,Kennard1927,Robertson1929,Dodonov2025} weight tails
quadratically and are saturated by Gaussians; $\Da$ is instead a first-moment ($L^1$)
functional of the rearranged density, finite whenever the density decays at least as fast as $|x|^{-2-\delta}$ for
some $\delta>0$, and, as shown above, \emph{not}
Gaussian-optimal. The \emph{entropic} relations~\cite{Hirschman1957,Beckner1975,BBM1975,
MaassenUffink1988,Wu2009,ColesRMP2017,Wang2019,Huang2021} we use rather than compete with: $\Da x\ge e^{h(x)-1}$
lower-bounds the width by an exponential of the entropy, but $\Da$ carries independent
geometric content (the two bounds saturate on different states). The \emph{measure /
support} principles of Donoho--Stark and Landau--Pollak%
~\cite{DonohoStark1989,LandauPollak1961,FollandSitaram1997,Benedicks1985,AmreinBerthier1977}
are the home of the confidence bound \eqref{eq:thm2}; the map
$\theta\mapsto\dx(\theta)=\Theta^{-1}(\theta)$ is the inverse of a measurable-set counterpart of the L\'evy
concentration function (the two coincide for unimodal densities) and $\Da x$ is its first moment. Finally $\Da$ is a
Lorenz/\emph{majorization} object: $\Theta$ is the Lorenz curve of $f$ and $\Da x$ is a
Schur-concave functional of the density, decreasing under majorization, which places
the relation in the universal-uncertainty framework of Partovi, Friedland--Gheorghiu--Gour
and Pucha\l a--Rudnicki--\.Zyczkowski~\cite{Partovi2011,FriedlandGheorghiuGour2013,
PuchalaRudnickiZyczkowski2013,MarshallOlkinArnold2011}.

There is also an information-theoretic reading.
The differential entropy $h(x)$ is a description length---the resolution-independent cost
of specifying $x$~\cite{CoverThomas2006}---and $e^{h(x)}$ is the effective support
length of $f$ (the typical-set size), so $\Da x\ge e^{h(x)-1}$ says that the average confidence width
is at least the effective support length divided by $e$. In finite dimension the
average confidence width is exactly Massey's \emph{guesswork}, the minimal expected
number of guesses needed to identify a measurement outcome~\cite{Massey1994,Arikan1996}, so
relations such as $W_zW_\phi\ge1.314$ bound the joint guessing cost in two complementary bases.
Possible applications include lossy quantization of wave functions, where the ratio of $\Da x$ to the
cell size sets the number of cells needed, and robust $L^1$ metrology, where
$\Da x=2\ev{|x-x_0|}$ is a tail-insensitive precision measure.

\section{Conclusion}\label{sec:conc}

The average confidence width is an operational, rearrangement-invariant measure of localization with a
clean analytic form, $\Da x=\int_0^\infty s f^\ast\dd s=\min_V\ev{V}$, and it obeys the scale-free
uncertainty relation $\Da x\,\Da p\ge c^\ast\hb$ for pure and mixed states. We proved
$\pi/e\le c^\ast\le E_0^2=1.2189803<4/\pi$: the optimal constant lies strictly below the Gaussian value,
and the candidate optimizer is the ground state of $|x|+|p|$---nodeless, self-dual, symmetric decreasing,
more sharply peaked than a Gaussian and with algebraic tails. The bound $E_0^2\hb$ holds rigorously whenever one
of the two densities is symmetric unimodal, by the sharp hybrid relations of Theorem~\ref{thm:hybrid}; the
second variation at $\psi_0$ is positive definite apart from the symmetry directions; the value $E_0^2$ is
approached, and never undercut, in extensive numerical minimizations including
expansions in coherent states; and the analogous statement in two dimensions is a theorem, with
Gaussian optimizers. What remains open is a proof of $c^\ast=E_0^2$ for states whose position and
momentum densities are both irregular, for which a new rearrangement inequality linking a function and its
Fourier transform appears to be needed. Finally, uncertainty relations for the ACW of complementary
observables contain the interferometric duality relations as special cases: exactly for the path qubit, and
through monotone relations with predictability and visibility for the double slit.

\begin{acknowledgments}
This work is supported by the National Natural Science Foundation of China
(Grant No.~12475020), the National Key Research and
Development Program of China (Grant No.~2023YFC2205802), and the Innovation
Program for Quantum Science and Technology (Grant No.~2021ZD0301701).
\end{acknowledgments}

\appendix

\section{The average confidence width: proofs}\label{app:props}

\begin{proof}[Proof of Proposition~\ref{prop:rearr}]
By the bathtub principle (Theorem~1.14 of Ref.~\cite{LiebLoss2001}), among all sets of measure $s$ the one maximizing
the captured probability is a super-level set of $f$, with maximum $\Theta(s)=\int_0^s f^\ast(u)\dd u$.
Since $\Theta$ is continuous and strictly increasing on $\supp f^\ast$ with $\Theta(0)=0$,
$\Theta(\infty)=1$, the minimal measure capturing $\theta$ is its inverse,
$\dx(\theta)=\Theta^{-1}(\theta)$. Therefore
$\Da x=\int_0^1\Theta^{-1}(\theta)\dd\theta=\int_0^\infty(1-\Theta(s))\dd s$, and
$\int_0^\infty(1-\Theta(s))\dd s=\int_0^\infty\!\!\int_s^\infty f^\ast(u)\dd u\,\dd s
=\int_0^\infty u f^\ast(u)\dd u$ by Fubini, which is the first equality of
\eqref{eq:rearr}. The change of variables $s=2|x|$ gives the $f^\sharp$ form.

For the variational formula, let $V\ge0$ satisfy $\mu\{V<t\}=t$ for all $t>0$ (then also $\mu\{V\le t\}=t$). By the layer-cake
representation,
\begin{align*}
  \int_\R Vf\dd x&=\int_0^\infty\!\Big(1-\int_{\{V\le t\}}f\dd x\Big)\dd t\\
  &\ge\ \int_0^\infty\bigl(1-\Theta(t)\bigr)\dd t=\Da x,
\end{align*}
because $\{V\le t\}$ is a set of measure $t$ and hence captures at most $\Theta(t)$. Equality holds if and only if
almost every sublevel set $\{V\le t\}$ is a maximizer in the bathtub problem, which is the case for the
rank potential $V_f$. As an infimum of functionals that are linear in $f$, $\Da x$ is concave in $f$, and
since $f=\bra x\rho\ket x$ is linear in $\rho$, it is concave in $\rho$. Finally $V(x)=2|x-x_0|$ is admissible,
which gives \eqref{eq:madineq}; equality requires the intervals $[x_0-t/2,x_0+t/2]$ to be super-level sets
of $f$ for almost every $t$, i.e.\ $f$ symmetric and non-increasing about $x_0$.
\end{proof}

\emph{Scaling.} Under $\psi(x)\mapsto\lambda^{1/2}\psi(\lambda x)$ one has
$f^\ast(s)\mapsto\lambda f^\ast(\lambda s)$, so $\Da x\mapsto\lambda^{-1}\Da x$, while the
induced momentum scaling gives $\Da p\mapsto\lambda\Da p$; hence $\Da x\,\Da p$ is dilation
invariant and \eqref{eq:form} is the only possible form.

\emph{Worked examples.} For a Gaussian $\psi\propto e^{-x^2/4\sigma^2}$,
$f^\ast(s)=\tfrac1{\sqrt{2\pi}\sigma}e^{-s^2/8\sigma^2}$ and
$\Da x=2\sqrt{2/\pi}\,\sigma\approx1.596\,\sigma=2\ev{|x|}$;
with $\sigma_p=\hb/2\sigma$, $\Da x\,\Da p=\tfrac4\pi\hb\approx1.273\,\hb$. For the two-sided
exponential amplitude $\psi=\lambda^{-1/2}e^{-|x|/\lambda}$ the rearranged density is exactly
exponential, $f^\ast(s)=\lambda^{-1}e^{-s/\lambda}$, and $\Da x=\lambda$; the momentum amplitude is a Lorentzian, so that
$g(p)=(2\lambda/\pi\hb)(1+\lambda^2p^2/\hb^2)^{-2}$ is symmetric unimodal with
$\Da p=2\ev{|p|}=4\hb/\pi\lambda$. Hence $\Da x\,\Da p=\tfrac4\pi\hb$, the same value as for the Gaussian.
For a box $f=\tfrac1{2D}\mathbf 1_{[-D,D]}$, $\Da x=D$ but the $\mathrm{sinc}^2$ momentum density
has a $\sim p^{-2}$ envelope, so $\Da p=\infty$. Finiteness of $\Da x\,\Da p$ requires both rearranged densities to
have a finite first moment, which for a non-increasing $f^\ast$ implies $f^\ast(s)=o(s^{-2})$ as $s\to\infty$.

\section{Proof of Proposition~\ref{cor:selfcontained}}\label{app:proofselfcontained}

\begin{proof}
Expand the square as $\theta_x\theta_p+(1-\theta_x)(1-\theta_p)-2\sqrt{\theta_x(1-\theta_x)}\sqrt{\theta_p(1-\theta_p)}$.
Using the involution $(\theta_x,\theta_p)\mapsto(1-\theta_x,1-\theta_p)$, which swaps
$\{\theta_x+\theta_p>1\}$ and $\{\theta_x+\theta_p<1\}$, one finds
$\iint_{>1}[\theta_x\theta_p+(1-\theta_x)(1-\theta_p)]=\tfrac{5}{24}+\tfrac{1}{24}=\tfrac14$,
while the symmetric weight $w(t)=\sqrt{t(1-t)}$ gives
$\iint_{>1}w(\theta_x)w(\theta_p) \dd \theta_x \dd \theta_p=\tfrac12\big(\int_0^1 w \dd t \big)^2=\tfrac12(\pi/8)^2=\pi^2/128$.
Thus $I=\tfrac14-2\cdot\tfrac{\pi^2}{128}=\tfrac14-\tfrac{\pi^2}{64}$.
\end{proof}

\section{The uncertainty relation: proofs}\label{app:relation}

Throughout this appendix $\ell=\hb=1$ unless stated otherwise.

\subsection{Entropic bound}

\begin{lemma}[Mean--entropy inequality]\label{lem:meanentropy}
If $r$ is a probability density on $[0,\infty)$ with mean $M=\int_0^\infty s\,r(s)\dd s<\infty$
and entropy $h[r]=-\int r\ln r \dd s $, then $h[r]\le1+\ln M$, with equality if and only if
$r(s)=M^{-1}e^{-s/M}$.
\end{lemma}
\begin{proof}
With $q(s)=M^{-1}e^{-s/M}$, nonnegativity of relative entropy gives
$0\le\int r\ln(r/q)=-h[r]+\ln M+\tfrac1M\int s\,r=-h[r]+\ln M+1$.
\end{proof}

\begin{proof}[Proof of Proposition~\ref{thm:entropic}]
The rearrangement $f^\ast$ is a density on $[0,\infty)$ with mean $\Da x$ by
\eqref{eq:rearr} and entropy $h[f^\ast]=h(x)$ (entropy is rearrangement invariant).
Lemma~\ref{lem:meanentropy} gives $h(x)\le1+\ln\Da x$, i.e.\ $\Da x\ge e^{h(x)-1}$;
likewise for momentum. Multiplying and inserting BBM,
$\Da x\,\Da p\ge e^{h(x)+h(p)-2}\ge e^{\ln(\pi e\hb)-2}=\tfrac\pi e\hb$. The BBM relation holds for
mixed states as well, because the differential entropy is a concave functional of the density.
\end{proof}

\subsection{Proof of Theorem~\ref{prop:hamiltonian}}

We use two facts about the operator $|p|$ (Secs.~7.12 and 7.17 of Ref.~\cite{LiebLoss2001}). First, its quadratic form has the representation
\begin{equation}\label{eq:pform}
  \ev{\psi||p||\psi}=\frac{1}{2\pi}\iint\frac{|\psi(x)-\psi(y)|^2}{(x-y)^2}\dd x\dd y ,
\end{equation}
equivalently $(|p|\psi)(x)=\pi^{-1}\,\mathrm{P.V.}\!\int[\psi(x)-\psi(y)](x-y)^{-2}\dd y$, and
$e^{-t|p|}$ has the strictly positive (Poisson) kernel $t/\{\pi[t^2+(x-y)^2]\}$. Second, symmetric
decreasing rearrangement does not increase the kinetic energy:
$\ev{\psi^\sharp||p||\psi^\sharp}\le\ev{\psi||p||\psi}$, where $\psi^\sharp$ denotes the symmetric
decreasing rearrangement of $|\psi|$.

\emph{(a) Spectrum and the dilation argument.} The quadratic form
$\mathfrak q[\psi]=\ev{|x|}+\ev{|p|}$ is nonnegative and closed on its natural domain, and defines the
self-adjoint operator $H=|x|+|p|$. Because both $|x|\to\infty$ and $|p|\to\infty$, the set
$\{\norm\psi\le1,\ \mathfrak q[\psi]\le\kappa\}$ is compact in $L^2$ for every $\kappa$ (Rellich's criterion, Theorem~XIII.65
of Ref.~\cite{ReedSimon4}), so $H$ has compact
resolvent and purely discrete spectrum (Theorem~XIII.64 of Ref.~\cite{ReedSimon4}); in particular $E_0=\min\mathfrak q$ is attained. For
\emph{every} state, \eqref{eq:madineq} gives the bound $\Da x\,\Da p\le4\ev{|x|}\ev{|p|}$, hence
$c^\ast\hb\le\inf_\psi4\ev{|x|}\ev{|p|}$. To
evaluate the right-hand side, restore $\ell$ and $\hb$: the dilation $\psi_\lambda(x)=\lambda^{1/2}\psi(\lambda x)$
sends $\ev{|x|}\mapsto\lambda^{-1}\ev{|x|}$ and $\ev{|p|}\mapsto\lambda\ev{|p|}$, so
\begin{equation*}
  \inf_\lambda\ev{\hat O}_\lambda
  =\inf_\lambda\!\Big(\frac{\ev{|x|}}{\lambda\ell}+\frac{\lambda\ell\,\ev{|p|}}{\hb}\Big)
  =2\sqrt{\frac{\ev{|x|}\,\ev{|p|}}{\hb}},
\end{equation*}
independent of $\ell$; minimizing over $\psi$ as well,
$E_0=\inf_\psi 2\sqrt{\ev{|x|}\ev{|p|}/\hb}$, i.e.\ $\inf_\psi4\ev{|x|}\ev{|p|}=E_0^2\hb$. At the minimizer the
two terms are equal, $\ev{|x|}/\ell=\ell\ev{|p|}/\hb=E_0/2$.

\emph{(b) No nodes, nondegeneracy.} The semigroup $e^{-t|p|}$ is positivity improving, because its kernel is
strictly positive, and $e^{-t|x|}$ is multiplication by a positive function. By the Trotter product formula
$e^{-tH}$ is positivity preserving, and by the perturbation theory of positivity-preserving semigroups it is
positivity improving (Theorems~XIII.43--XIII.45 in Sec.~XIII.12 of Ref.~\cite{ReedSimon4}). Since $E_0$ is an
eigenvalue, the Perron--Frobenius theorem (Theorem~XIII.44 of Ref.~\cite{ReedSimon4}) states that $E_0$ is
nondegenerate and that the ground state can be chosen strictly positive: $\psi_0>0$, i.e.\ $\psi_0$ has no nodes.

The mechanism behind this theorem is elementary and worth displaying, since it is specific to the nonlocal
operator $|p|$. Eigenfunctions can be chosen real, and for a real ground state $\psi$ the inequality
$\big||\psi(x)|-|\psi(y)|\big|\le|\psi(x)-\psi(y)|$ in \eqref{eq:pform} shows
$\mathfrak q[|\psi|]\le\mathfrak q[\psi]$, so $u=|\psi|\ge0$ is also a ground state and satisfies
$|p|u=(E_0-|x|)u$. If $u(x_0)=0$ at some point, the right-hand side vanishes there, whereas (for a
sufficiently regular $u$)
\begin{equation*}
  (|p|u)(x_0)=-\frac1\pi\int\frac{u(y)}{(x_0-y)^2}\dd y<0 ,
\end{equation*}
a contradiction: a nonnegative function is strictly ``superharmonic'' for $|p|$ at each of its zeros. Hence a
ground state cannot vanish anywhere, has a fixed sign, and two independent ground states cannot exist, since
the combination $\psi_a(0)\psi_b-\psi_b(0)\psi_a$ would vanish at the origin.

\emph{Even, self-dual, symmetric decreasing.} $H$ commutes with the parity $\Pi$ and with the Fourier
transform $\Fou$ (which exchanges $|x|$ and $|p|$). By nondegeneracy $\Pi\psi_0=\pm\psi_0$, and positivity selects
$+$. Likewise $\Fou\psi_0=\zeta\psi_0$ with $\zeta^2=1$ because $\Fou^2=\Pi$; taking the scalar product with the
Gaussian $\phi_G=\pi^{-1/4}e^{-x^2/2}=\Fou\phi_G$ gives $\ip{\phi_G}{\psi_0}=\zeta\ip{\phi_G}{\psi_0}$ with
$\ip{\phi_G}{\psi_0}>0$, so $\zeta=1$ and $\tilde\psi_0=\psi_0$. Finally, by Proposition~\ref{prop:rearr},
$\int|x||\psi_0^\sharp|^2\le\int|x||\psi_0|^2$, and by the rearrangement inequality for $|p|$ quoted above
$\ev{\psi_0^\sharp||p||\psi_0^\sharp}\le\ev{\psi_0||p||\psi_0}$; thus $\psi_0^\sharp$ is a positive ground state, and
by uniqueness $\psi_0^\sharp=\psi_0$. Therefore $f=\psi_0^2$ and $g=\tilde\psi_0^2=f$ are symmetric
decreasing, \eqref{eq:madineq} holds with equality in both representations, and
$\Da x\,\Da p[\psi_0]=4\ev{|x|}\ev{|p|}=E_0^2\hb$. (That $\psi_0$ is \emph{strictly} decreasing on $(0,\infty)$
follows from the representation of $-\psi_0'$ obtained in the proof of Proposition~\ref{prop:local} below.)

\emph{(c) Tails.} The semigroup $e^{-t|p|}$ is the transition semigroup of the Cauchy process, the symmetric $1$-stable
L\'evy process, whose jump density $\nu(x)=1/(\pi x^2)$ can be read off from the Poisson kernel; the potential $|x|$ is locally
bounded, tends to infinity, and is comparable on unit intervals ($\sup_{|y-x|\le1}|y|\le3\inf_{|y-x|\le1}|y|$ for $|x|\ge2$).
The two-sided ground-state estimates of Kaleta and L\H{o}rinczi (Corollary~2.2 and Example~4.8 of
Ref.~\cite{KaletaLorinczi2015}) therefore apply and give $\psi_0(x)\asymp\nu(x)/|x|$ for large $|x|$, i.e.\
$c_1(1+|x|)^{-3}\le\psi_0(x)\le c_2(1+|x|)^{-3}$ for all $x$, since $\psi_0$ is positive and bounded on compact sets.

\emph{Asymptotic forms.} The leading terms in Eq.~\eqref{eq:asympt} follow from a formal expansion, in which the
remainders are assumed to be of lower order. In momentum space the eigenvalue equation reads
$(|p|-E_0)\tilde\psi_0(p)=-\widetilde{(|x|\psi_0)}(p)$. The large-$p$ behavior of the right-hand side is
governed by the kink of $|x|\psi_0(x)\simeq\psi_0(0)|x|$ at the origin; since the Fourier transform of $|x|$
is $-2/(\sqrt{2\pi}\,p^2)$ away from $p=0$, we get $\tilde\psi_0(p)\simeq 2\psi_0(0)/(\sqrt{2\pi}|p|^3)$, and by
self-duality the same tail for $\psi_0(x)$. Transforming the $|p|^{-3}$ tail back gives the singular
term $\pi^{-1}\psi_0(0)\,x^2\ln|x|$ at the origin. Expanding $(|p|-E_0)^{-1}$ gives the leading correction $\tilde\psi_0(p)\simeq C|p|^{-3}(1+E_0/|p|+\dots)$. The numerical solution confirms both statements: the fitted coefficient of
$x^2\ln|x|$ is $0.264$ against $\psi_0(0)/\pi=0.260$, and at $|x|=20$ the lattice solution, extrapolated in the
lattice size, agrees with the corrected tail to better than $1\%$.

\subsection{Proof of Theorem~\ref{thm:hybrid}}

Consider first a pure state. Multiplying $\psi$ by $e^{-ip_0x/\hb}$ shifts $g$ by $p_0$ and leaves $f$
unchanged, so we may take $p_0=0$; we may also assume $\ev{|p|}<\infty$. With $u=|\psi|$, Eq.~\eqref{eq:pform} gives
$\ev{u||p||u}\le\ev{\psi||p||\psi}$, and the rearrangement inequality gives
$\ev{\psi^\sharp||p||\psi^\sharp}\le\ev{u||p||u}$. On the other hand
$\Da x[\psi]=\Da x[\psi^\sharp]=2\ev{|x|}_{\psi^\sharp}$ by Proposition~\ref{prop:rearr}. Therefore
\begin{equation*}
  \Da x[\psi]\cdot2\ev{|p|}_\psi\ \ge\ 4\ev{|x|}_{\psi^\sharp}\ev{|p|}_{\psi^\sharp}\ \ge\ E_0^2\hb
\end{equation*}
by Theorem~\ref{prop:hamiltonian}(a). The second relation in \eqref{eq:hybrid} follows by exchanging the roles of
$x$ and $p$. For a mixed state with spectral decomposition $\rho=\sum_iq_i\ket{\psi_i}\bra{\psi_i}$, concavity gives
$\Da x[\rho]\ge\sum_iq_i\alpha_i$ with $\alpha_i=\Da x[\psi_i]$, while $2\ev{|p-p_0|}_\rho=\sum_iq_i\beta_i$ with
$\beta_i=2\ev{|p-p_0|}_{\psi_i}$; since $\alpha_i\beta_i\ge E_0^2\hb$, the Cauchy--Schwarz inequality yields
$(\sum_iq_i\alpha_i)(\sum_iq_i\beta_i)\ge(\sum_iq_i\sqrt{\alpha_i\beta_i})^2\ge E_0^2\hb$. Equality in the pure-state chain
requires $\psi^\sharp$ to be a dilate of $\psi_0$; $|\psi|$ to be a translate of $\psi^\sharp$, by the strict
form of Riesz's rearrangement inequality (Theorem~3.9 of Ref.~\cite{LiebLoss2001}) applied to the Gaussian kernel of $e^{-tp^2}$ in the
representation $|p|=(4\pi)^{-1/2}\int_0^\infty(1-e^{-tp^2})\,t^{-3/2}\dd t$; and $\psi/|\psi|$ to be constant
(after the boost by $p_0$), by \eqref{eq:pform}, since $\psi^\sharp>0$. For a mixed state, equality would require every
eigenvector $\psi_i$ to saturate the pure-state relation and, in the concavity step, all the densities $|\psi_i|^2$ to share
the same minimizing potential $V$, hence the same center; Cauchy--Schwarz then fixes a common dilation, so that all $\psi_i$
would coincide up to phases, which is impossible for orthogonal vectors. Hence equality forces $\rho$ to be pure.

\subsection{Proof of Proposition~\ref{prop:local}}

For sufficiently small $\eps$ the densities of $\psi_\eps$ are unimodal: $f_0=\psi_0^2$ is strictly decreasing
in $|x|$ with $f_0''\to-\infty$ at the origin, its derivative decays only algebraically, and the perturbation is smooth
and rapidly decreasing; the same holds in momentum space. For a unimodal density with mode $c$ and level
sets $\{f>t\}=(x_-(t),x_+(t))$ one has $\Da x=\tfrac12\int(x_+-x_-)^2\dd t$ and
$2\ev{|x-c|}=\int[(x_+-c)^2+(x_--c)^2]\dd t$, whence the exact identity
\begin{equation}\label{eq:gain}
  2\ev{|x-c|}-\Da x=2\int_0^{\max f}\big[m(t)-c\big]^2\dd t ,
\end{equation}
with $m=\tfrac12(x_++x_-)$ the midpoint of the level set.
Write $f_\eps=f_0+\eps\,\delta f+O(\eps^2)$ with $\delta f=2\psi_0\,\mathrm{Re}\,\eta$. If $f_0(x_t)=t$, $x_t>0$, the endpoints of the
level set are $x_\pm=\pm x_t\pm\eps\,\delta f(\pm x_t)/|f_0'(x_t)|$, so the midpoint is displaced by
$m(t)=\eps\,\delta f_o(x_t)/|f_0'(x_t)|$, where $\delta f_o=2\psi_0\eta_1$ is the odd part, while the length changes only through the even
part. Inserting into \eqref{eq:gain} and changing variables from $t$ to $x_t$,
\begin{align*}
  2\ev{|x|}-\Da x&=2\eps^2\!\int_0^\infty\!\frac{\delta f_o^2}{|f_0'|}\dd x+o(\eps^2)\\
  &=2\eps^2\!\int_\R\frac{\psi_0}{|\psi_0'|}\,\eta_1^2\dd x+o(\eps^2);
\end{align*}
the integral converges because, by the asymptotic forms \eqref{eq:asympt}, $\psi_0/|\psi_0'|\simeq\pi/(2|x||\ln|x||)$ at the origin and
$\simeq|x|/3$ at infinity. (Identity \eqref{eq:gain} holds with $c$ replaced by $0$ for all levels
$t<f_\eps(0)$, since it only requires $c$ to lie in the level set; the remaining levels contribute $o(\eps^2)$.) The same computation in momentum space involves the real part $\chi_2$ of the odd part of
$\tilde\eta$. The second-order expansion of $4\ev{|x|}\ev{|p|}$ is elementary: with $\ev{\psi_0|\eta}=0$ and
$b=\mathrm{Re}\ev{\psi_0||x||\eta}=-\mathrm{Re}\ev{\psi_0||p||\eta}$ one finds
$E_0^2+\eps^2\{2E_0\ev{\eta|(H-E_0)|\eta}-16b^2\}$, in which the even and odd parts of $\eta$ decouple, only the
even part contributes to $b$, and $\ev{\eta_2|(H-E_0)|\eta_2}=\ev{\chi_2|(H-E_0)|\chi_2}$ by Fourier invariance.
Collecting terms gives \eqref{eq:secondvar}. $Q_e\ge0$ because $4\ev{|x|}\ev{|p|}\ge E_0^2$ for all states. On
even functions orthogonal to $\psi_0$ the operator $2E_0(H-E_0)$ is positive definite (by nondegeneracy and discreteness),
and $Q_e$ differs from it by a rank-one term that involves only $\mathrm{Re}\,\eta_e$; hence the null directions of $Q_e$ form
at most a one-dimensional real space. It is spanned by the dilation mode $(\tfrac12+x\,\dd/\dd x)\psi_0$, along which
$4\ev{|x|}\ev{|p|}$ is constant.

It remains to show $L\ge0$ on odd functions. Differentiating $H\psi_0=E_0\psi_0$ and using
$[\dd/\dd x,|x|]=\mathrm{sgn}\,x$ gives $(H-E_0)\psi_0'=-\mathrm{sgn}(x)\psi_0$, and since
$\psi_0'=-\mathrm{sgn}(x)|\psi_0'|$, this is $L\psi_0'=0$. Identify odd functions with functions on
$(0,\infty)$. There the semigroup $e^{-t|p|}$ has the kernel
$\tfrac t\pi\{[t^2+(x-y)^2]^{-1}-[t^2+(x+y)^2]^{-1}\}>0$, so it is positivity improving, and so is the semigroup
generated by $H$ on the odd sector. Since the ground state of $H$ is nondegenerate and even, the spectrum of $H$ on
the odd sector lies strictly above $E_0$ (numerically it starts at $E_1=2.2323$), so that $(H-E_0)^{-1}$ is
positivity improving there; hence
$-\psi_0'=(H-E_0)^{-1}[\mathrm{sgn}(x)\psi_0]$ is strictly positive on $(0,\infty)$. By the asymptotic forms
\eqref{eq:asympt}, the potential in
\eqref{eq:Lop} is locally bounded on $(0,\infty)$, confining at infinity (it grows like $2|x|/3$), and its singularity at the
origin, of order $1/(|x|\,|\ln|x||)$, is infinitesimally form bounded relative to $|p|$, because for odd $\eta$ Eq.~\eqref{eq:pform} and the positive
definiteness of the kernel $(x+y)^{-2}=\int_0^\infty te^{-t(x+y)}\dd t$ give the Hardy-type inequality
$\ev{\eta||p||\eta}\ge\pi^{-1}\!\int\eta^2/|x|\dd x$. Thus $L$ is a well-defined operator with
compact resolvent generating a positivity-improving semigroup on $L^2(0,\infty)$; its ground state is
nondegenerate and positive, and no other eigenfunction can be positive (it would have to be orthogonal to the ground
state). Since $-\psi_0'>0$ is an eigenfunction with eigenvalue $0$, it is the ground state: $L\ge0$ with
$\ker L=\mathrm{span}\{\psi_0'\}$. A lattice discretization of $L$ with $N=4096$ points gives the odd-sector
eigenvalues $0.002$, $1.108$, $1.926$, $2.578$, which extrapolate to $0$, $1.105$, $1.923$, $2.575$ for
$N\to\infty$; the lowest eigenvector has squared overlap $0.99997$ with $\psi_0'$.
We also checked \eqref{eq:secondvar} directly by evaluating $J$ at finite $\eps$ on a lattice of $2^{16}$ points. For example, for the normalized
perturbation $\eta\propto xe^{-x^2/2}$ the predicted coefficient of $\eps^2$ is $0.2785$ (it would be $2.545$ without
the rearrangement gain), and $(J[\psi_\eps]-J[\psi_0])/\eps^2=0.314,\,0.288,\,0.279$ for $\eps=0.2,\,0.1,\,0.05$; for
$\eta\propto\psi_0'$ and for the boost $\eta\propto ix\psi_0$ the same ratio tends to zero. Python script is available in~\cite{SM}.

\subsection{Proof of Proposition~\ref{prop:mixed}}

(a) Proposition~\ref{prop:rearr} and Lemma~\ref{lem:meanentropy} are statements about densities, and BBM holds
for mixed states by concavity of the entropy; Theorem~\ref{thm:hybrid} was proved for mixed states above.
For the confidence bound, let $X,Y$ be sets of finite measure with spectral projections $P_X$ (of $x$)
and $Q_Y$ (of $p$). The pairs $(\Tr\rho P_X,\Tr\rho Q_Y)$ attainable with mixed states form the convex hull of
the pairs attainable with pure states; but the latter set is the numerical range of $P_X+iQ_Y$, which is
convex by the Toeplitz--Hausdorff theorem~\cite{GustafsonRao1997}. Hence any pair $(\theta_x,\theta_p)$ captured by $(X,Y)$ in a mixed
state lies in the closure of the pairs captured by the same sets in pure states (the closure takes care of
mixtures of infinitely many pure states), and \eqref{eq:thm2}, which only involves
$\mu(X)\mu(Y)$ and is continuous in $(\theta_x,\theta_p)$, extends to mixed states; so does Proposition~\ref{cor:selfcontained}.
(b) Concavity was shown in Appendix~\ref{app:props}. With $\alpha_i=\Da x[\rho_i]$, $\beta_i=\Da p[\rho_i]$,
$\Da x[\rho]\Da p[\rho]\ge(\sum q_i\alpha_i)(\sum q_i\beta_i)\ge(\sum q_i\sqrt{\alpha_i\beta_i})^2$ by the Cauchy--Schwarz inequality.
(c) Insert the Frank--Lieb relation~\cite{FrankLieb2012} $h(x)+h(p)\ge\ln(2\pi\hb)+S(\rho)$ into
$\Da x\,\Da p\ge e^{h(x)+h(p)-2}$.

\subsection{Comb states and two dimensions}\label{app:comb}

\emph{Proof of Proposition~\ref{prop:2d}.} In $d$ dimensions the bathtub argument gives
$\Da^{(d)}x=\int_0^\infty sf^\ast(s)\dd s$ with $f^\ast$ the decreasing rearrangement by volume, and
Lemma~\ref{lem:meanentropy} gives $\Da^{(d)}x\ge e^{h(x)-1}$; with BBM in $d$ dimensions this is
\eqref{eq:ddim}. For $\psi\propto e^{-x\cdot Ax/2}$ the density $f\propto e^{-x\cdot Ax}$ has elliptical super-level sets
$\{x\cdot Ax<r^2\}$ of volume $\omega_dr^d/\sqrt{\det A}$, so for $d=2$ the rearranged density is
$f^\ast(s)=M^{-1}e^{-s/M}$ with $M=\pi/\sqrt{\det A}$; likewise $g^\ast$ is exponential with mean
$\pi\hb^2\sqrt{\det A}$. Both mean--entropy inequalities and BBM are then equalities, and the product is $\pi^2\hb^2$.

\emph{Comb states.} For $\delta\ll a\ll\Sigma$ the position density of the comb state of
Sec.~\ref{sec:numerics} is $f(x)\simeq a\,\nu(y)\,\tau(u)$ with $x=y+u$, $y\in a\mathbb{Z}$ the tooth position,
$\nu$ the normalized envelope density and $\tau$ the normalized tooth density, and the measure is
$\dd x\simeq\dd y\dd u/a$. Its rearrangement is therefore that of the two-dimensional product density
$\nu\otimes\tau$ with all volumes divided by $a$, i.e.\ $\Da x=\Da^{(2)}[\nu\otimes\tau]/a$. For Gaussian $\nu$
and $\tau$ with standard deviations $\Sigma/\sqrt2$ and $\delta/\sqrt2$ this gives $\Da x=\pi\Sigma\delta/a$.
By the Poisson summation formula the momentum wave function is again a comb, with spacing
$2\pi\hb/a$, tooth width $\hb/\Sigma$ and envelope width $\hb/\delta$, so
$\Da p=\pi(\hb/\Sigma)(\hb/\delta)/(2\pi\hb/a)=\hb a/2\Sigma\delta$ and the product is $\pi\hb/2$. For
a general envelope amplitude $\mathcal E(y)$ and tooth amplitude $\mathcal T(u)$, the two densities are those of the two-dimensional product state
$\mathcal E\otimes\mathcal T$ in position and momentum space (up to the exchange of the factors), hence in the comb limit
$\Da x\,\Da p=\Da^{(2)}x\,\Da^{(2)}p/(2\pi\hb)\ge\pi\hb/2$ by Proposition~\ref{prop:2d}. The value $\pi/2$ is
approached only asymptotically; the comb-like local minima found numerically in finite variational spaces lie slightly above it.

\section{Double slit}\label{app:slit}

\emph{Position.} Let $\Theta_1$ be the concentration function of a single slit. For disjoint slits
$\Theta(s)=\max_{s_1+s_2=s}[w_1\Theta_1(s_1)+w_2\Theta_1(s_2)]$. If $w_1\ge w_2$, the maximizer has
$s_1\ge s_2$, and increasing $w_1$ at the expense of $w_2$ does not decrease $\Theta(s)$ for any $s$ and increases
it for some; since $\Da x=\int_0^\infty(1-\Theta)\dd s$, it is strictly decreasing in $\mathcal P=w_1-w_2$. At $\mathcal P=0$
the two slits contribute equally, $f^\ast(s)=\tfrac12f_1^\ast(s/2)$, and $\Da x=2\Da x^{(1)}$. For a uniform
slit of width $b$, $f^\ast$ equals $w_1/b$ on $[0,b)$ and $w_2/b$ on $[b,2b)$, so
$\Da x=\tfrac b2(w_1+3w_2)=(2-\mathcal P)\,\Da x^{(1)}$ with $\Da x^{(1)}=b/2$.

\emph{Momentum.} All statements refer to the many-fringe regime, in which the envelope $G=|\tilde\varphi|^2$ is constant over a fringe
period, so that the density is a product of the envelope and the fringe profile on each period and the majorization
arguments below apply period by period. The measure of $\{g>t\}$ is $\mu_g(t)=\int m_{\mathcal V}(t/G(p))\dd p$, where
$m_{\mathcal V}(r)=\pi^{-1}\arccos[(r-1)/\mathcal V]$ for $|r-1|<\mathcal V$ (and $1$ or $0$ below or above) is the fraction of a period on which
$1+\mathcal V\cos u>r$. Using $\Da p=\tfrac12\int_0^\infty\mu_g(t)^2\dd t$, a uniform envelope of width $B$ gives
$\Da p=\tfrac B2\int_0^{1+\mathcal V}m_{\mathcal V}^2\dd r=\tfrac B2(1-4\mathcal V/\pi^2)$, since $\int_{-1}^1\arccos^2z\,\dd z=\pi^2-4$. In
general, for $\mathcal V<\mathcal V'$ the density $g_{\mathcal V}$ is the mixture $(1-\mathcal V/\mathcal V')g_0+(\mathcal V/\mathcal V')g_{\mathcal V'}$, and $g_0=G$ is strictly majorized
by $g_{\mathcal V'}$; concavity of $\Da$ then gives $\Da p[g_{\mathcal V}]>\Da p[g_{\mathcal V'}]$. The lower bound $R_1>\tfrac12$ holds
because the fringe profile $1+\cos u$ is strictly majorized by the two-valued profile equal to $2$ on half a
period and $0$ on the other half, for which $\Da p=\Da p^{(1)}/2$. For a Gaussian envelope $R_1=0.705$.

\section{Numerical methods}\label{app:num}

All computations use $\hb=1$ and a uniform grid of $N$ points with spacing $\Delta x$; the
momentum grid is the one induced by the fast Fourier transform, $\Delta p=2\pi/(N\Delta x)$, with
$\tilde\psi(p)=(2\pi)^{-1/2}\!\int\psi(x)e^{-ipx}\dd x$; the self-dual choice $\Delta x=\Delta p=\sqrt{2\pi/N}$ is used
unless a finer resolution in $x$ is needed. For a discretized density the
decreasing rearrangement is the sorted array $f_{(1)}\ge f_{(2)}\ge\cdots$, and
$\Da x=\tfrac{\Delta x^2}{2}\sum_k(2k-1)f_{(k)}$; the rank potential is $V_j=\Delta x(r_j-\tfrac12)$ with $r_j$ the rank of $f_j$. A single Python script that reproduces every number and
figure of this paper is provided as Supplemental Material~\cite{SM}.

\emph{The ground state of $|x|+|p|$.} (i) On the self-dual lattice the operator is
$\mathrm{diag}(|x|)+\Fou^{-1}\mathrm{diag}(|p|)\Fou$, and its lowest eigenvalue is computed by a Lanczos method.
The trapezoidal rule misses the kinks of $|x|$ and $|p|$ at the origin, which produces the leading lattice error
$E_0(N)=E_0-(2\pi/3)\psi_0(0)^2/N+\dots$; the measured coefficient $-1.3947$ agrees with
$-(2\pi/3)\psi_0(0)^2=-1.39474$. Richardson extrapolation of $N=2^{13},2^{14},2^{15}$ gives
$E_0=1.10407442$. (The value $1.1032$ quoted in an earlier version of this work was an unextrapolated lattice
value.) (ii) In the basis of even Hermite functions the matrix elements of $|x|$ are computed by Gauss--Legendre
quadrature, and those of $|p|$ follow from $\Fou\ket n=(-i)^n\ket n$. The Rayleigh--Ritz eigenvalues are
upper bounds (up to the rounding and quadrature errors of the matrix elements, which are far below the quoted digits),
$E_0\le1.1040747$ with $350$ functions, and decrease like $K^{-2}$; extrapolation gives
$1.10407442$. We quote $E_0=1.1040744(1)$. The first excited (odd) level is $E_1=2.2323$ and the second even level is $2.7728$.

\emph{Minimization.} In a variational space spanned by an orthonormalized basis $\{\phi_k\}$ (Hermite
functions; coherent states on the phase-space lattice $\sqrt\pi\,\mathbb{Z}^2$ inside a disc of radius $R$,
i.e.\ twice the von Neumann density; or the $N$ lattice sites themselves), one SCF step consists of (1)
evaluating the rank potentials $V,\tilde V$ of the current state on the grid, (2) forming the matrices of $V(x)$
and $\tilde V(p)$ in the basis, and (3) replacing the state by the ground state of $\lambda V+\tilde V/\lambda$ with
$\lambda$ optimized by a one-dimensional search. Each step decreases $J$, and the iteration stops when the
decrease is below $10^{-12}$ (typically after $10$--$40$ steps). The gradient used in the quasi-Newton runs is
$\partial J/\partial\psi^\ast=\Da p\,(V_f-\Da x)\psi+\Da x\,\Fou^{-1}(\tilde V_g-\Da p)\Fou\psi$. The fidelity with
$\psi_0$ is maximized over translations, boosts, and dilations. On the full lattice $\mathbb Z_N$ we minimize the
same mid-point discretization of the continuum functional, which in terms of the discrete widths $W$ of
Definition~\ref{def:spread} reads $(2\pi/N)(W_x-\tfrac12)(W_p-\tfrac12)$. This product has
lattice corrections of relative order $1/N$ (from the mid-point rule for the rearranged density and from the
position of the state relative to the lattice sites), and the finite problem has additional exact symmetries
(the maps $j\mapsto uj$ mod $N$ with $u$ odd permute both bases), so that unstructured starting points such as
white noise may converge to scrambled copies of a minimizer, with the same rearranged densities; we therefore
compare final states through their rearranged densities rather than through fidelities. For $N=256$ and $512$
the smallest products found, $1.2204864$ and $1.2197329$, are attained by lattice copies of $\psi_0$ (the
concentration curves $\Theta$ agree with those of $\psi_0$ within $5\times10^{-3}$ and $2\times10^{-3}$) and
agree to within $3\times10^{-6}$ with the estimate $E_0^2+(\pi/6)E_0\psi_0(0)^2/N$ for a copy of $\psi_0$ placed optimally with
respect to the lattice. For $N=128$ such copies give $1.2219953$, and one run found the slightly lower value
$1.2206715$ at a comb state whose teeth are a single lattice site wide in both representations; such
lattice-scale states are not resolved by the grid and have no continuum counterpart. All values lie above $E_0^2$.

\end{document}